\documentclass[cmp,numbook,envcountsame]{svjour}  
\usepackage{amsmath}
\usepackage{amsfonts,amssymb,mathtools}

\journalname{Communications in Mathematical Physics}

\newcommand{\cqfd}{\hfill $\square$}

\newcommand{\RR}{\mathbb{R}}

\newcommand{\NN}{\mathbb{N}}
\newcommand{\CC}{\mathbb{C}}

\newcommand{\cH}{\mathcal{H}}
\newcommand{\cF}{\mathcal{F}}
\newcommand{\cD}{\mathcal{D}}
\newcommand{\cK}{\mathcal{K}}
\newcommand{\cB}{\mathcal{B}}
\newcommand{\cE}{\mathcal{E}}
\newcommand{\cJ}{\mathcal{J}}

\newcommand{\cC}{\mathcal{C}}

\newcommand{\cU}{\mathcal{U}}
\newcommand{\cM}{\mathcal{M}}

\begin{document}

\title{{Asymptotics in Spin-Boson type models}}
\titlerunning{Asymptotics in Spin Boson type models}

\author{Thomas Norman Dam\inst{1} Jacob Schach M\o ller \inst{2}}
\institute{Aarhus Universitet, Nordre Ringgade 1, 8000 Aarhus C  Denmark.\\ \email{tnd@math.au.dk} \and
 Aarhus Universitet, Nordre Ringgade 1, 8000 Aarhus C  Denmark.\\ \email{jacob@math.au.dk}}
\authorrunning{Thomas Norman Dam, Jacob Schach M\o ller}

\date{}
\communicated{}

\maketitle
\begin{abstract}
In this paper, we investigate a family of models for a qubit interacting with a bosonic field. More precisely, we find asymptotic limits of the Hamiltonian as the strength of the interaction tends to infinity. The main result has two applications. First of all, we show that self-energy renormalisation schemes similar to that of the Nelson model will never give a physically interesting result. This is because any limit obtained through such a scheme would be independent of the qubit. Secondly, we find that exited states exist in the massive Spin-Boson models for sufficiently large interaction strengths. We are also able to compute the asymptotic limit of many physical quantities.
\end{abstract}

\section{Introduction}
In this paper, we consider a family of models for a qubit coupled to a bosonic field, which we will call spin-boson type models. These models has been investigated in many papers, so many properties are well known. Asymptotic completeness along with basic spectral properties were discussed in \cite{DerezinskiGerard} and \cite{Marcel}. Existence and regularity of ground states were discussed in \cite{Volker}, \cite{Gerard}, \cite{Hasler} and \cite{Hirokawa2}. Furthermore, properties at positive temperature were discussed in \cite{Merkli} and \cite{Moller}.

One of the main ingredients in the papers \cite{Volker}, \cite{Thomas1} and \cite{Hasler} is the so-called spin-parity symmetry. In the paper \cite{Thomas1}, this symmetry is used to decompose the Hamiltonian into two so-called fiber Hamiltonians, which are both perturbations of Van Hove Hamiltonians. This symmetry is also essential to the analysis conducted in this paper, and we will need the results from \cite{Thomas1}.

To avoid a full technical description in the introduction, we will specialise to the 3-dimensional Spin-Boson model. In this case the bosons have dispersion relation $\omega(k)=\sqrt{m^2+\lVert k \lVert^2}$ with $m\geq 0$ and $k\in \RR^3$. The interaction between the field and the qubit is parametrised by the functions
\begin{align*}
v_{g,\Lambda}(k)=g\frac{\chi_\Lambda(  \omega(k) )}{\sqrt{\omega(k)}}
\end{align*}
where $\{ \chi_\Lambda \}_{\Lambda\in (0,\infty)}$ is a family of ultraviolet functions such that $v_{g,\Lambda}\in \cD(\omega^{-1/2})$. We will assume that $\Lambda\mapsto \chi_\Lambda(k)$ is an increasing function and
\begin{align*}
\lim_{\Lambda\rightarrow \infty} \chi_\Lambda(k)=1
\end{align*}
for all $k\in \RR^3$. Let $2\eta>0$ be the size of the energy gap in the qubit and $H_{g,\Lambda,\eta}$ be the Hamiltonian of the full system. Then we show the following:
\begin{enumerate}
\item First we consider self-energy renormalisation schemes. In such schemes one defines $f_{g,\eta}(\Lambda)= \inf(\sigma(H_{g,\Lambda,\eta}))$ and proves that $\{ H_{g,\Lambda,\eta}-f_{g,\eta}(\Lambda)\}_{\Lambda\in (0,\infty)}$ converges in strong or norm resolvent sense to an operator $H^{\textup{Ren}}_{g,\eta}$ as $\Lambda$ tends to $\infty$. Using Corollary \ref{KorUVREN} and Lemma \ref{Expkonv} below we see:
\begin{equation*}
\lim_{\Lambda\rightarrow \infty} (f_{g,\eta}(\Lambda)+\lVert \omega^{-1/2}1_{\{\omega>1  \}}v_{g,\Lambda} \lVert^2)= \lVert \omega^{-1/2}1_{\{\omega<1  \}}v_{g,\Lambda} \lVert^2
\end{equation*}
 which is independent of $\eta$ and
\begin{equation*}
(H_{g,\Lambda,\eta}+\lVert \omega^{-1/2}1_{\{\omega>1  \}}v_{g,\Lambda} \lVert^2+i)^{-1}-(H_{g,\Lambda,0}+\lVert \omega^{-1/2}1_{\{\omega>1  \}}v_{g,\Lambda} \lVert^2+i)^{-1}
\end{equation*}
converges to 0 in norm as $\Lambda$ tends to $\infty$. From this we conclude that if a self-energy renormalisation scheme exists then $H^{\textup{Ren}}_{g,\eta}$ must be independent of $\eta$, which is not physically interesting. In other words, the contribution from the qubit disappears, as the ultraviolet cutoff is removed. This result is similar to the result in \cite{DirkPizzo}, where it is shown, that the mass-shell in a certain model becomes "almost flat" as the ultraviolet cutoff is removed. So the contribution from the matter particle vanishes as the ultraviolet cutoff is removed.

\item  If $m>0$ we can take $g$ to infinity instead. In this case the result yields that an exited state exists for $g$ very large. Furthermore, the energy difference between the exited state and the ground state converges to 0. Taking $g$ to infinity is not a purely mathematical exercise as experiments can go beyond the ultra deep coupling regime. This was achieved by Yoshihara, K. et al. and published in Nature Physics  \cite{Yos}.
\end{enumerate}
We will also prove two smaller results. The first result is about regularity of ground states with respect to the number operator. The result only applies to the infrared regular case, but is close to optimal and extends the results found in \cite{Hirokawa2}. The second result is a condition under which the massive spin-boson model has an exited state in the mass gap.

\section{Notation and preliminaries}
We start by fixing notation. If $X$ is a topological space we will write $\cB(X)$ for the Borel $\sigma$-algebra. Furthermore if $(\cM,\cF,\mu)$ is a measure space we will for $1\leq p\leq \infty$ write $L^p(\cM,\cF,\mu)$ for the corresponding $L^p$ space.

Throughout this paper $\cH$ will denote the state space of a single boson which we will assume to be a separable Hilbert space. Let $S_n$ denote projection of $\cH^{\otimes n}$ onto the subspace of symmetric tensors. The bosonic (or symmetric) Fock space is defined as
\begin{equation*}
\cF_b(\cH)=\bigoplus_{n=0}^\infty S_n( \cH^{\otimes n}).
\end{equation*}
If $\cH=L^2(\cM,\cF,\mu)$ where $(\cM,\cF,\mu)$ is a $\sigma$-finite measure space then  $S_n(\cH^{\otimes n})=L_{sym}^2(\cM^{n},\cF^{\otimes n},\mu^{\otimes n})$. An element $\psi\in \cF_b(\cH)$ is an infinite sequence of elements which is written as $\psi=(\psi^{(n)})$. We also define the vacuum $\Omega=(1,0,0,\dots)$. Furthermore, we will write
\begin{equation*}
S_n(f_1\otimes\cdots\otimes f_n)=f_1\otimes_s\cdots \otimes_s f_n.
\end{equation*}
For $g\in \cH$ one defines the annihilation operator $a(g)$ and creation operator $a^{\dagger}(g)$ on symmetric tensors in $\cF_b(\cH)$ by $a(g)\Omega=0,a^\dagger(g)\Omega=g$ and
\begin{align*}
a(g)( f_1\otimes_s\cdots\otimes_s f_n )&=\frac{1}{\sqrt{n}}\sum_{i=1}^{n} \langle g,f_i \rangle f_1\otimes_s\cdots\otimes_s \hat{f}_i\otimes_s\cdots\otimes_s f_n\\
a^\dagger(g)( f_1\otimes_s\cdots\otimes_s f_n )&=\sqrt{n+1}g\otimes_s f_1\otimes_s\cdots\otimes_s f_n
\end{align*}
where $\widehat{f}_i$ means that $f_i$ is omitted from the tensor product. One can show that these operators extends to closed operators on $\cF_b(\cH)$ and that $(a(g))^*=a^{\dagger}(g)$. Furthermore we have the canonical commutation relations which states
\begin{equation*}
\overline{[a(f),a(g)]}=0=\overline{[a^\dagger(f),a^\dagger(g)]} \,\,\text{and}\,\,\, \overline{[a(f),a^\dagger(g)]}=\langle f,g\rangle.
\end{equation*}
One now introduces the selfadjoint field operators
\begin{equation*}
\varphi(g)=\overline{ a(g)+a^\dagger(g) }.
\end{equation*}
Let $\omega$ be a selfadjoint and non-negative operator on $\cH$ with domain $\cD(\omega)$. Write $(1\otimes)^{k-1} \omega(\otimes 1)^{n-k}$ for the operator $B_1\otimes...\otimes B_n$ where $B_k=\omega$ and $B_j=1$ if $j\neq k$. We then define the second quantisation of $\omega$ to be the selfadjoint operator
\begin{equation}\label{Sumdecomp}
d\Gamma(\omega)=0\oplus \bigoplus_{n=1}^{\infty} \left(\sum_{k=1}^{n} (1\otimes)^{k-1} \omega(\otimes 1)^{n-k}\right)\biggl \lvert_{S_n(\cH^{\otimes n})}.
\end{equation}
If $\omega$ is a multiplication operator then $d\Gamma(\omega)$ acts on elements in $S_n(\cH^{\otimes n})$ as multiplication by $\omega_n(k_1,\dots,k_n)=\omega(k_1)+\cdots+\omega(k_n)$. The number operator is defined as $N=d\Gamma(1)$. Let $U$ be unitary from $\cH$ to $\cK$. Then we define the unitary from $\cF_b(\cH)$ to $\cF_b(\cK)$ by
\begin{equation*}
\Gamma(U)=1\oplus \bigoplus_{n=1}^\infty  U^{\otimes n}\mid_{S_n(\cH^{\otimes n})},
\end{equation*}
For $n\in \NN_0=\NN\cup \{0\}$ we also define the operators $d\Gamma^{(n)}(\omega)=d\Gamma(\omega)\mid_{S_n(\cH^{\otimes n})}$ and $\Gamma^{(n)}(U)=\Gamma(U)\mid_{S_n(\cH^{\otimes n})}$. See \cite{Thomas1} for a proof of the following lemma:
\begin{lemma}\label{Lem:SecondQuantisedProp}
	Let $\omega$ be a selfadjoint and non negative operator on $\cH$ and let $m=\inf(\sigma(\omega))$. For $n\geq 1$ we have
	\begin{align*}
	\sigma(d\Gamma^{(n)}(\omega) )&= \overline{ \{  \lambda_1+\cdots+\lambda_n\mid \lambda_i\in \sigma(\omega) \}},\\
	\inf(\sigma(d\Gamma^{(n)}(\omega)))&= nm.
	\end{align*}
	Furthermore, $d\Gamma(\omega)$ will have compact resolvents if and only if  $\omega$ has compact resolvents. Also, $d\Gamma^{(n)}(\omega)$ is injective for $n\geq 1$ if $\omega$ is injective.
\end{lemma}
We now introduce the Weyl representation. For any $g\in \cH$ we define the corresponding exponential vector
\begin{equation*}
\epsilon(g)=\sum_{n=0}^{\infty} \frac{g^{\otimes n}}{\sqrt{n!}}.
\end{equation*}
One may prove that if $\cD\subset \cH$ is dense then the set $\{ \epsilon(f)\mid f\in \cD \}$ is a linearly independent total subset of $\cF_b(\cH)$. Let $\cU(\cH)$ be the unitaries from $\cH$ into $\cH$. Fix now $U\in \cU(\cH)$ and $h\in \cH$. The corresponding Weyl transformation is the unique unitary map $W(h,U)$ satisfying
\begin{equation*}
W(h,U)\epsilon(g)=e^{-\lVert h\lVert^2/2-\langle f,Ug \rangle}\epsilon(h+Ug).
\end{equation*}
for all $g\in \cH$. One may easily check that $(h,U)\mapsto W(h,U)$ is strongly continuous. Furthermore one may check the relation
\begin{equation}\label{eq:Weylcomp}
W(h_1,U_1)W(h_2,U_2)=e^{-i\text{Im}(\langle h_1,U_1h_2 \rangle)}W((h_1,U_1)(h_2,U_2)),
\end{equation}
where $(h_1,U_1)(h_2,U_2)=(h_1+U_1h_2,U_1U_2)$. If $\omega$ is selfadjoint and $f\in \cH$ then we have
\begin{align}\label{grupper1}
e^{itd\Gamma(\omega)}&=\Gamma(e^{it\omega})=W(0,e^{it\omega})\\ \label{grupper2}
e^{it\varphi(if)}&=W(tf,1).
\end{align}
The following lemma is important and well known (see e.g \cite{Hirokawa1} and \cite{Derezinski}):
\begin{lemma}\label{Lem:FundamentalIneq}
	Let $\omega\geq 0$ be a selfadjoint, non negative and injective operator on $\cH$. If $v \in \cD(\omega^{-1/2})$ then $\varphi(v)$ is $d\Gamma(\omega)^{1/2}$ bounded. In particular, $\varphi(v)$ is $N^{1/2}$ bounded. We have the following bound
	\begin{equation*}
	\lVert \varphi(v) \psi \lVert\leq 2 \lVert (\omega^{-1/2}+1)v \lVert  \lVert (d\Gamma(\omega)+1)^{1/2}\psi \lVert   
	\end{equation*}
	which holds on $\cD(d\Gamma(\omega)^{1/2})$. In particular $\varphi(v)$ is infinitesimally $d\Gamma(\omega)$ bounded. Furthermore, $\sigma (d\Gamma(\omega)+\varphi(v))= -\lVert \omega^{-1/2}v \lVert^2+\sigma(d\Gamma(\omega))$. 
\end{lemma}

\section{The Spin-Boson model}
Let $\sigma_x,\sigma_y,\sigma_z$ denote the Pauli matrices
\begin{align*}
\sigma_x=\begin{pmatrix}
0 & 1 \\ 1& 0
\end{pmatrix} \,\,\,\,\, \sigma_y=\begin{pmatrix}
0 & -i \\ i& 0
\end{pmatrix} \,\,\,\,\, \sigma_z=\begin{pmatrix}
1 & 0 \\ 0& -1
\end{pmatrix}
\end{align*}
and define $e_1=(1,0)$ and $e_{-1}=(0,1)$. The total system has the Hamiltonian 
\begin{equation*}
H_{\eta}(v,\omega):=\eta \sigma_z\otimes 1+1\otimes d\Gamma(\omega)+\sigma_x\otimes \varphi(v),
\end{equation*}
which is here parametrised by $ v\in \cH, \eta\in \CC$ and $\omega$ selfadjoint on $\cH$. We will also need the fiber operators:
\begin{equation*}
F_{\eta}(v,\omega)=\eta\Gamma(-1)+d\Gamma(\omega)+\varphi(v).
\end{equation*}
acting in $\cF_b(\cH)$. If the spectra are real we define
\begin{align*}
E_{\eta}(v,\omega)&:=\inf(\sigma(H_{\eta}(v,\omega)))\\
\cE_{\eta}(v,\omega)&:=\inf(\sigma(F_{\eta}(v,\omega))).
\end{align*}
For $\omega$ selfadjoint on $\cH$ we define
\begin{equation*}
m(\omega)=\inf\{ \sigma(\omega) \} \,\,\,\,\, \text{and}\,\,\,\,\, m_{\textup{ess}}(\omega)=\inf\{ \sigma_{\textup{ess}}(\omega) \}.
\end{equation*}
Standard perturbation theory and Lemma \ref{Lem:FundamentalIneq} yields:
\begin{proposition}\label{Lem:BasicPropertiesSBmodel} Let$\omega\geq 0$ be a selfadjoint, non negative and injective operator on $\cH$, $v\in \cD(\omega^{-1/2})$ and $\eta\in \CC$. Then the operators $F_{\eta}(v,\omega)$ and $H_\eta(v,\omega)$ are closed on the respective domains
	\begin{align*}
	\cD(F_{\eta}(v,\omega))&=\cD(d\Gamma(\omega))\\
	\cD(H_{\eta}(v,\omega))&=\cD(1\otimes d\Gamma(\omega)).
	\end{align*}
	If $\cD$ is a core of $\omega$ then the linear span of the following sets
	\begin{align*}
	\cJ(\cD)&:=\{\Omega \}\cup \bigcup_{n=1}^\infty \{  f_1\otimes_s\cdots \otimes_s f_n\mid f_j\in \cD \}\\ \widetilde{\cJ}(\cD)&:=\{ f_1\otimes f_2\mid f_1\in \{ e_1,e_{-1} \} ,f_2\in \cJ(\cD)  \}
	\end{align*}
	is a core for $F_{\eta}(v,\omega)$ and  $H_{\eta}(v,\omega)$ respectively. Furthermore, both operators are selfadjoint and semibounded if $\eta\in \RR$.
\end{proposition}
\noindent From the paper \cite{Thomas1} we find the following theorem:
\begin{theorem}\label{Thm:Spectral Theory of decomposition}
	Let $\phi=(\phi_1,\phi_{-1})=e_1\otimes \phi_1+e_{-1}\otimes\phi_{-1}$ be an element in $ \cF_b(\cH)^2=\cF_b(\cH)\oplus \cF_b(\cH)\approx \CC^2\otimes \cF_b(\cH)$. Write $\phi_i=(\phi^{(k)}_i)$ for $i\in \{-1,1\}$. Let $i\in \{-1,1\}$. Define $\widetilde{\phi}_{i}=(\widetilde{\phi}^{(k)}_i)$ where
	\begin{equation*}
	\widetilde{\phi}^{(k)}_i=\begin{cases} \phi^{(k)}_i & \text{k is even}\\ \phi^{(k)}_{-i} & \text{k is odd} 
	\end{cases}
	\end{equation*}
	and $V(\phi_1,\phi_{-1})=(\widetilde{\phi}_1,\widetilde{\phi}_{-1})$. Then
	\begin{enumerate}
	\item [\textup{(1)}]  $V$ is unitary with $V^*=V$.
	
	\item [\textup{(2)}] If $\omega\geq 0$ be a selfadjoint, non negative and injective operator on $\cH$ then $V1\otimes d\Gamma(\omega)V^*=1\otimes d\Gamma(\omega)$. Furthermore, if $\eta\in \RR$ and $v\in \cD(\omega^{-1/2})$ then 
	\begin{equation*}
		V H_\eta(v,\omega) V^*=F_{-\eta }(v,\omega)\oplus F_{\eta}(v,\omega).
	\end{equation*}
	
	\item [\textup{(3)}] Let $\omega\geq 0$ be a selfadjoint, non negative and injective operator on $\cH$, $\eta\in \RR$ and $v\in \cD(\omega^{-1/2})$. Then $E_{\eta}(v,\omega)=\cE_{-\lvert \eta\lvert}(v,\omega)$ and $H_{\eta}(v,\omega)$ has a ground state if and only if the operator $F_{-\lvert \eta\lvert}(v,\omega)$ has a ground state. This is the case if $m(\omega)>0$, and it is non degenerate if $\eta\neq0$. Furthermore,
		\begin{align*}
		\inf(\sigma_{\textup{ess}}(F_{\lvert \eta\lvert }(v,\omega)))&=\cE_{-\lvert \eta\lvert }(v,\omega)+m_{\textup{ess}}(\omega)\\
		\inf(\sigma_{\textup{ess}}(H_{ \eta }(v,\omega)))&=E_{\eta}(v,\omega)+m_{\textup{ess}}(\omega)
		\end{align*}
		and $\cE_{\lvert \eta\lvert }(v,\omega)>\cE_{-\lvert \eta\lvert }(v,\omega)$ if and only if both $\eta\neq0$ and $m(\omega) \neq 0$.
	
	\item [\textup{(4)}] Let $\omega\geq 0$ be a selfadjoint, non negative and injective operator on $\cH$, $\eta\in \RR$ and $v\in \cD(\omega^{-1/2})$. If $\phi$ is a ground state for $H_\eta(v,\omega)$ then
	\begin{align*}
	V\phi=\begin{cases}
	e_{-\textup{sign}(\eta)}\otimes \psi & \eta\neq 0\\
	e_{-1}\otimes \psi_{-1}+e_{1}\otimes \psi_1 & \eta=0
	\end{cases}
	\end{align*}
	where $\psi$ is a ground state for $F_{-\lvert \eta\lvert}(v,\omega)$ and $\psi_1,\psi_{-1}$ are either 0 or a ground state for $F_{0}(v,\omega)$.
\end{enumerate}
\end{theorem}

\section{Results}
In this section we state the results which are proven in this paper. Throughout this section $\omega$ will always denote an injective, non negative and selfadjoint operator on $\cH$. Furthermore, we will write $m=m(\omega)$ and $m_{\textup{ess}}=m_{\textup{ess}}(\omega)$. The main technical result is the following theorem:
\begin{theorem}\label{Centralconv}
Let $\{ v_g \}_{g\in (0,\infty)}\subset \cD(\omega^{-1/2})$ and $P_{\omega}$ denote the spectral measure corresponding to $\omega$. For each $\widetilde{m}>0$ we define $P_{\widetilde{m}}= P_{\omega}( (\widetilde{m},\infty))$ and $\overline{ P}_{\widetilde{m}}=1-P_{\widetilde{m}}=P_\omega([0,\widetilde{m}])$. Assume that there is $\widetilde{m}>0$ such that:
\begin{enumerate}
\item [\textup{(1)}] $\{\overline{ P}_{\widetilde{m}} v_g\}_{g\in (0,\infty)}$ converges to $v\in \cD(\omega^{-1/2})$ in the graph norm of $\omega^{-1/2}$.

\item [\textup{(2)}]$\lVert \omega^{-1}P_{\widetilde{m}} v_g\lVert $ diverges to $\infty$ as $g$ tends to infinity.
\end{enumerate}
 Then the family of operators given by
\begin{align}\label{eq:LIGHED}
\widetilde{F}_{\eta,\widetilde{m}}(v_g,\omega):&=W(\omega^{-1} P_{\widetilde{m}} v_g,1)F_{\eta}(v_g,\omega)W(\omega^{-1} P_{\widetilde{m}} v_g,1)^*+\lVert  \omega^{-1/2} P_{\widetilde{m}} v_g \lVert^2\nonumber \\&=\eta W(2\omega^{-1} P_{\widetilde{m}} v_g,-1)+d\Gamma(\omega)+\varphi(\overline{ P}_{\widetilde{m}} v_g)
\end{align}
converges to $d\Gamma(\omega)+\varphi(v)$ in norm resolvent sense as $g$ tends to $\infty$. Furthermore, $\{ \widetilde{F}_{\eta,\widetilde{m}}(v_g,\omega) \}_{g\in (0,\infty)}$ is bounded below by $-\lvert \eta\lvert-\sup_{g\in (0,\infty)}\lVert \omega^{-1/2}\overline{ P}_{\widetilde{m}}  v_g \lVert^2>-\infty $. 
\end{theorem} 
\noindent The assumption in part (1) is critical. Divergence where $\omega$ is small can lead to problems. This is proven in Proposition \ref{Counter} below.

In the strongly coupled Spin-Boson model one usually has $v_g=g\widetilde{v}$ where $\widetilde{v}\in \cD(\omega^{-1/2})$ and $g\in (0,\infty)$ is the strength of the interaction. We can now answer what happens as $g $ goes to $\infty$.
\begin{corollary}\label{KorGSFiber}
	Let $v\in  \cH$, $\eta\in \RR$ and assume $m>0$. Then there exists $g_0>0$ such that $\cE_{\eta}(gv,\omega)$ is a non degenerate eigenvalue of $F_{\eta}(gv,\omega)$ when $g>g_0$. Furthermore, one may pick a family of normalised vectors $\{ \psi_g \}_{g\in [g_0,\infty )}$ such that $g\mapsto  \psi_g$ is smooth, $F_{\eta}(gv,\omega)\psi_g=\psi_g \cE_{\eta}(gv,\omega)$ and
	\begin{align*}
		\lim\limits_{g\rightarrow \infty} \lVert \psi_g-e^{-g^2\lVert \omega^{-1}v \lVert^2 } \epsilon(-g\omega^{-1}v)\lVert&=0,\\
		\lim\limits_{g\rightarrow \infty}  \frac{ \langle \psi_g,N\psi_g \rangle -g^2\lVert \omega^{-1}v \lVert^2}{g} &=0,\\\lim_{g\rightarrow \infty} ( \cE_{\eta}(gv,\omega) +g^2\lVert \omega^{-1/2}v \lVert^2 ) &=0.
	\end{align*}
	If $\eta <0$ then $g\mapsto \cE_{\eta}(gv,\omega) +g^2\lVert \omega^{-1/2}v\lVert$ is strictly increasing.
\end{corollary} 

\begin{corollary}\label{KorGSTotal}
	Let $v\in  \cH$ and $\eta\in \RR$. If $m>0$ there is $g_0>0$ such that $H_\eta(gv,\omega)$ has an exited state with energy $\widetilde{E}_\eta(gv,\omega)$ for $g>g_0$. Furthermore
	\begin{equation*}
	\lim_{g\rightarrow \infty} ( E_\eta(gv,\omega)-\widetilde{E}_\eta(gv,\omega) )=0.
	\end{equation*}
\end{corollary}
\begin{corollary}\label{KorUVREN}
	Assume $\cH=L^2(\cM,\cF,\mu)$ and $\omega$ is a multiplication operator on this space. Let $v:\cM \rightarrow \CC$ is measurable and $\{ \chi_g \}_{g\in (0,\infty)}$ be a collection of functions from $\RR$ into $[0,1]$. Assume $g\mapsto \chi_g(x)$ is increasing and converges to 1 for all $x\in \RR$. Assume furthermore that $k\mapsto \chi_g(\omega(k))v(k)\in \cD(\omega^{-1/2})$ and that there is $\widetilde{m}>0$ such that $\widetilde{v}:=1_{\{ \omega\leq \widetilde{m} \}}v\in \cD(\omega^{-1/2})$. If $k\mapsto  \omega(k)^{-1}v(k)1_{\{\omega>1\}}(k)\notin \cH $ there are unitary maps $\{U_g\}_{g\in (0,\infty)}$ and $\{V_g\}_{g\in (0,\infty)}$ independent of $\eta$ such that:
		\begin{enumerate}
		\item[\textup{(1)}] $\{ U_gF_{\eta}(v_g,\omega)U_g^*+\lVert \omega^{-1/2}1_{\{\omega>\widetilde{m} \}}v_g\lVert^2\}_{g\in (0,\infty)}$ converges in norm resolvent sense to the operator $d\Gamma(\omega)+\varphi(\widetilde{v})$ as $g$ tends to infinity.
		
		\item[\textup{(2)}] $\{V_gH_{\eta}(v_g,\omega)V_g^*+\lVert \omega^{-1/2}1_{\{\omega>\widetilde{m} \}}v_g\lVert^2\}_{g\in (0,\infty)}$ is uniformly bounded below and converges in norm  resolvent sense to the operator
		\begin{align*}
		\widetilde{H}:=(d\Gamma(\omega)+\varphi(\widetilde{v}))\oplus(d\Gamma(\omega)+\varphi(\widetilde{v}))
		\end{align*}
		as $g$ tends to $\infty$. This implies
				\begin{equation*}
				(H_{\eta}(v_g,\omega)+\lVert \omega^{-1/2}1_{\{\omega>\widetilde{m} \}}v_g\lVert^2+i)^{-1}- (H_{0}(v_g,\omega)+\lVert \omega^{-1/2} 1_{\{\omega>\widetilde{m} \}} v_g\lVert^2+i)^{-1}
				\end{equation*}
				will converge to 0 in norm as $g$ tends to $\infty$.
		\end{enumerate}

\end{corollary}
To prove a result similar to Corollary $\ref{KorGSFiber}$ in the massless case one needs to work a bit harder. First we shall need
\begin{theorem}\label{BasicpropGS}
	Assume $\cH=L^2(\cM,\cK,\nu)$ and $\omega$ is multiplication by a measurable function. Let $v\in  \cD(\omega^{-1/2})$, $g\in (0,\infty)$ and $\eta\leq 0$. Assume that $F_{\eta}(gv,\omega)$ has a ground state $\psi_{g,\eta}=( \psi_{g,\eta}^{(n)} )$. Then
	\begin{enumerate}
		\item[\textup{(1)}] We may choose $\psi_{g,\eta }$ such that $\psi_{g,\eta }^{(0)}>0$ and $(-1)^n\overline{v}^{\otimes n}\psi_{g,\eta }^{(n)}>0$ almost everywhere on $\{ v\neq 0 \}^n$.

		\item[\textup{(2)}] Almost everywhere the following inequality holds
		\begin{equation*}
		\lvert \psi_{g,\eta }^{(n)}(k_1,\dots,k_n)\lvert \leq \frac{g^n}{\sqrt{n!}} \frac{\lvert v(k_1)\lvert \cdots\lvert v(k_n)\lvert }{\omega(k_1)\cdots\omega(k_n)}.
		\end{equation*}
		In particular, $\psi_{g,\eta }^{(n)}$ is zero outside $\{ v\neq 0 \}^n$ almost everywhere.
		
		\item[\textup{(3)}] Assume $v\in \cD(\omega^{-1})$, $f:\NN_0\rightarrow [0,\infty)$ is a function and assume $F_{\eta}(gv,\omega)$ has a ground state for all $\eta\leq 0$. Then $H_{a}(gv,\omega)$ has a ground state $\phi_{g,a}$ for all $a\in \RR$ and we have
		\begin{align*}
		\alpha_{g,f,v,\omega}:=\sum_{n=0}^{\infty} \frac{f(n)^2g^{2n} \lVert \omega^{-1}v \lVert^{2n}  }{n!}<\infty &\iff \psi_{g,\eta}\in \cD(f(N)) \,\,\, \forall \eta\leq 0\\&\iff \phi_{g,a}\in \cD(1\otimes f(N)) \,\,\, \forall a\in \RR
		\end{align*}
		In particular, $\psi_{g,\eta}\in \cD(\sqrt[p]{N!} )$ and $\phi_{g,\eta}\in \cD(1\otimes \sqrt[p]{N!} )$ for all $p>2$.
		\end{enumerate}
\end{theorem}
This extends the result which was proven using path measures in \cite{Hirokawa2}. Similar point wise estimates can also be found in \cite{Frlich2}. In the last two results we will assume $\cH=L^2(\RR^\nu,\cB(\RR^\nu),\lambda_{\nu})$ where $\lambda_{\nu}$ is the Lebesgue measure. Furthermore, we assume $\omega$ is a multiplication operator on $\cH$.
\begin{theorem}\label{Mssless-Case}
	Let $v\in  \cD(\omega^{-1})$ and $\eta\leq 0$. Then there is a family $\{ \psi_g \}_{g\in \RR }$ of normalised eigenstates for $F_\eta(gv,\omega)$ and
	\begin{align*}
	&\lim\limits_{g\rightarrow \infty} ( \cE_{\eta}(gv,\omega) +g^2\lVert \omega^{-1/2}v \lVert^2 )=0.\\
	&\lim\limits_{g\rightarrow \infty} \frac{ \langle \psi_g,N\psi_g \rangle-g^2\lVert \omega^{-1}v \lVert^2 }{g^2 }=0.
	\end{align*}
\end{theorem} 
\noindent The following is a simple criterion for the existence of an exited state in the massive Spin-Boson model.  
\begin{theorem}\label{Dim1+2}
	Assume $m>0$ and
	\begin{equation}\label{eq:Assumption}
	\int_{\RR^\nu}\frac{\lvert v(k)\lvert^2}{\omega(k)-m}dk=\infty.
	\end{equation}
	Then both $F_{ \eta}(v,\omega)$ and $F_{ -\eta }(v,\omega)$ have a ground state and $H_\eta(gv,\omega)$ will have an excited state. The condition is satisfied if $\omega \in C^2(\RR^\nu,\RR)$, $\nu\leq 2$ and there is $x_0\in \RR^\nu$ such that $\omega(x_0)=m$ and $\lvert v\lvert $ is bounded from below by a positive number on a ball around $x_0$. This holds for the physical model with $\nu\leq 2$.
\end{theorem}

\section{Proof of the main technical result}
In this section we shall investigate operators of the form
\begin{equation*}
\widetilde{F}_{\eta}(v,\omega):= d\Gamma(\omega)+\eta W(v,-1)
\end{equation*}
indexed by $\eta\in \RR$, $v\in \cH$ and $\omega$ selfadjoint and non negative on $\cH$. 
\begin{proposition}\label{Redfiber}
Assume $\eta\in \RR$, $v\in \cH$ and $\omega$ is selfadjoint, non negative and injective on $\cH$. Then $\widetilde{F}_{\eta}(v,\omega)$ is selfadjoint on $\cD(d\Gamma(\omega))$. Furthermore $\widetilde{F}_{\eta}(v,\omega)$ is bounded from below by $-\lvert \eta\lvert $ and $\widetilde{F}_{\eta}(v,\omega)$ has compact resolvents if $\omega$ has compact resolvents. 
\end{proposition}
\begin{proof}
Using equation $(\ref{eq:Weylcomp})$ we see
\begin{equation*}
W(v,-1)W(v,-1)=e^{-i\text{Im}(\langle v,-v \rangle)}W(v-v,(-1)^2)=W(0,1)=1
\end{equation*}
so $W(v,-1)=W(v,-1)^{-1}=W(v,-1)^*$ since $W(v,-1)$ is unitary. Hence $\widetilde{F}_{\eta}(v,\omega)=d\Gamma(\omega)+\eta W(v,-1)$ is selfadjoint on $\cD(d\Gamma(\omega))$. Furthermore, the lower bound follows from Lemma \ref{Lem:SecondQuantisedProp} and the fact that  $-1\leq  W(v,-1)\leq 1$. If $\omega$ has compact resolvents, then so does $d\Gamma(\omega)$ by Lemma \ref{Lem:SecondQuantisedProp} and hence
	\begin{equation*}
	(F_{\eta}(v,\omega)+i)^{-1}=(d\Gamma(\omega)+i)^{-1}+\eta(d\Gamma(\omega)+i)^{-1}W(v,-1)(F_{\eta}(v,\omega)+i)^{-1}
	\end{equation*}
	will be compact.\cqfd
\end{proof}

\begin{lemma}\label{WeakConv}
	Assume $\{ v_g \}_{g\in (0,\infty)}$ is a collection of elements in $\cH$ such that $\lVert v_g \lVert $ diverges to $\infty$. Then $W(v_g,-1)$ converges weakly to 0 as $g$ goes to $\infty$.
\end{lemma}
\begin{proof}
	By \cite[Theorem 4.26]{Weidmann} it is enough to check a dense subset. By linearity it is enough to check a set that spans a dense set. Hence it is enough to check exponential vectors $\epsilon(g)$ for any $g\in \cH$. We calculate
	\begin{align*}
	\langle \epsilon(g_1),W(v_g,-1) \epsilon(g_2) \rangle&=e^{-\lVert v_g\lVert^2/2+\langle v_g,g_2 \rangle}\langle \epsilon(g_1),\epsilon(v_g-g_2) \rangle\\&=e^{-\lVert v_g\lVert^2/2+\langle v_g,g_2 \rangle+\langle g_1,v_g \rangle-\langle g_1,g_2 \rangle},
	\end{align*}
	which converges to 0.\cqfd
\end{proof}
The following Lemma contains all the technical constructions we need. The techniques goes back to Glimm and Jaffe (see \cite{GlimmJaffe}) but has also been used in \cite{Thomas1}.
\begin{lemma}\label{MainTech}
Assume $\omega$ is a selfadjoint, non negative and injective operator on $\cH$. Let $P_\omega$ be the spectral measure of $\omega$ and $\widetilde{m}>0$. Define the measurable function $f_k:\RR\rightarrow \RR$
	\begin{equation*}
	f_k(x)= x1_{(0,\widetilde{m}]}(x) + \sum_{n=0}^{\infty} (n+1)2^{-k} 1_{(n2^{-k},(n+1)2^{-k}]\cap (\widetilde{m},\infty)}(x).
	\end{equation*}	
	along  with $\omega_k=\int_{\RR} f_k(\lambda)dP_{\omega}(\lambda)$. Then the following holds
	\begin{enumerate}
		\item[\textup{(1)}] $\widetilde{F}_{\eta}(v,\omega_k)$ converges to $\widetilde{F}_{\eta}(v,\omega)$ in norm resolvent sense uniformly in $v$.

		\item[\textup{(2)}] Let $\{ v_g \}_{g\in (0,\infty)}$ be a collection of elements in $ P_\omega((\widetilde{m},\infty))\cH$. For each $k\in \mathbb{N}$, there are Hilbert spaces $\cH_{1,k},\cH_{2,k}$, selfadjoint operators $\omega_{1,k},\omega_{2,k}\geq 0$, a collection of elements $\{ \widetilde{v}_{g,k} \}_{g\in (0,\infty)}\subset \cH_{1,k}$ and a collection of unitary maps $\{ \cU_{g,k} \}_{g\in (0,\infty)}$ such that
		\begin{align*}
		\cU_{g,k}:\cF_b(\cH)\rightarrow \cF_b(\cH_{1,k})\oplus \left (\bigoplus_{n=1}^\infty \cF_b(\cH_{1,k})\otimes S_n((\cH_{2,k})^{\otimes n})\right),  
		\end{align*}
		 $\omega_{1,k}\geq 2^{-k}$ has compact resolvents, $\lVert v_g\lVert=\lVert \widetilde{v}_{k,g}\lVert$ for all $g>0$ and
		\begin{align*}
		\cU_{g,k}\widetilde{F}_{\eta}(v_g,\omega_k)\cU_{g,k}^*=&\widetilde{F}_{\eta}(\widetilde{v}_{g,k},\omega_{1,k})\\&\oplus \bigoplus_{n=1}^\infty \left(\widetilde{F}_{(-1)^n\eta }(\widetilde{v}_{g,k},\omega_{1,k})\otimes 1+1\otimes d\Gamma^{(n)}(\omega_{2,k})\right)
		\end{align*}
		for all $\eta\in \RR$.
			\end{enumerate}
	
\end{lemma}
\begin{proof}
	(1): We may pick a $\sigma$-finite measure space $(\cM,\cF,\mu)$ and a unitary map $U:\cH\rightarrow L^2(\cM,\cF,\mu)$ such that $\widetilde{\omega}=U\omega U^*$ is multiplication by a strictly positive and measurable map. Conjugation with the unitary map $\Gamma(U)$, Lemma \ref{Lem:SeconduantisedBetweenSPaces} and $U\omega_kU^*=f_k(U\omega_kU^*)$ gives us
	\begin{align*}
	 \lVert  (\widetilde{F}_{\eta}&(v,\omega_k)-\xi)^{-1}-(\widetilde{F}_{\eta}(v,\omega)-\xi)^{-1} \lVert \\&=	\lVert (\widetilde{F}_{\eta}(Uv,f_k(U\omega U^*))-\xi)^{-1}-(\widetilde{F}_{\eta}(Uv,U\omega U^*)-\xi)^{-1} \lVert
	\end{align*}
	for all $\xi\in \RR\backslash \CC$. Hence we may assume $\omega$ is multiplication by a strictly positive map, which we shall also denote $\omega$. Using standard theory for the spectral calculus (see \cite{Schmudgen}) we find $\omega_k$ is multiplication by $\omega_k(x):=f_k(\omega(x))$. Write $\omega=\omega_\infty$ and note that $\omega_k>0$ for all $k\in \NN\cup \{ \infty \}$ by construction. Furthermore,
	\begin{equation}\label{estimaa:Omega}
	\sup_{x\in \cM}\left\lvert  \frac{\omega_k(x)-\omega(x)}{\omega(x)} \right\lvert\leq \frac{2^{-k}}{\widetilde{m}}.
	\end{equation}
	 Now $d\Gamma^{(n)}( \omega_k )$ acts on $L_{\text{sym}}^2(\cM^n,\cF^{\otimes n},\mu^{\otimes n})$ like multiplication with the map
	\begin{equation*}
	\omega^{(n)}_k(x_1,\dots,x_n)=\omega_k(x_1)+\cdots+\omega_k(x_n)
	\end{equation*}
	for all $k\in \NN\cup \{ \infty \}$. Equation (\ref{estimaa:Omega}) gives that $\lvert \omega^{(n)}(x)-\omega_k^{(n)}(x) \lvert\leq 2^{-k}\widetilde{m}^{-1}\omega^{(n)}(x)$ for all $x\in \cM^n$ so $\cD(d\Gamma^{(n)}( \omega ))\subset \cD(d\Gamma^{(n)}( \omega_k ))$ for all $k\in \NN$ and $n\in \mathbb{N}_0$. Furthermore we find for $\psi \in \cD(d\Gamma^{(n)}( \omega ))$ that
	\begin{equation*}
	\lVert (d\Gamma^{(n)}( \omega)-d\Gamma^{(n)}( \omega_k)) \psi \lVert \leq  \widetilde{m}^{-1}2^{-k}\lVert d\Gamma^{(n)}( \omega) \psi \lVert.
	\end{equation*} 
	 So for all $\psi \in \cD(d\Gamma(\omega))$ we have $\psi \in \cD(d\Gamma(\omega_k))$ and
	\begin{equation*}
	\lVert (\widetilde{F}_{\eta}(v,\omega) -\widetilde{F}_{\eta}(v,\omega_k) ) \psi\lVert=\lVert (d\Gamma( \omega)-d\Gamma( \omega_k)) \psi\lVert \leq \frac{2^{-k}}{\widetilde{m}}\lVert d\Gamma(\omega)\psi\lVert.
	\end{equation*}
	Let $\varepsilon>0$ and $\xi\in \CC\backslash \RR $. We now estimate
	\begin{align*}
	\lVert ((\widetilde{F}_\eta(v,\omega)+\xi   )^{-1}&-(\widetilde{F}_\eta(v,\omega_k)+\xi   )^{-1})\psi \lVert\\&\leq  \frac{1}{\lvert \text{Im}(\xi)\lvert  }\lVert      (d\Gamma(\omega)-d\Gamma(\omega_k))(\widetilde{F}_\eta(v,\omega)+\xi   )^{-1}\psi \lVert\\& \leq \frac{1}{\lvert \text{Im}(\xi)\lvert  }\frac{2^{-k}}{\widetilde{m}} \lVert d\Gamma(\omega )(\widetilde{F}_\eta(v,\omega) +\xi  )^{-1}\widetilde{\psi} \lVert\\& \leq  \frac{1}{\lvert \text{Im}(\xi)\lvert  }\frac{2^{-k}}{\widetilde{m}} \left( 1+\frac{1}{\lvert \text{Im}(\xi)\lvert  }+ \frac{\lvert \xi \lvert }{\lvert \text{Im}(\xi)\lvert} \right)\lVert \psi \lVert 
	\end{align*}
	which shows norm resolvent convergence uniformly in $v$.

	(2): For each $k\in \mathbb{N}$ we define
	\begin{equation*}
	C_k=\bigg \{ c\in \mathbb{N}_0 \bigg \lvert P_{c,k}:=P_{\omega}((\widetilde{m},\infty) \cap (c2^{-k},(c+1)2^{-k} ] )\neq 0 \bigg \}
	\end{equation*}
	 For each $c\in C_k$ let $\cK_{c,k}$ be a Hilbert space with dimension $\dim(P_{c,k})-1$. In case this number is infinity we pick a Hilbert space we countably infinite dimension. Define $\cK=P_{\omega}([0,\widetilde{m}])\cH$ and note that $\cK$ reduces  $\omega$. Define the spaces
	 \begin{equation*}
	\cH_{1,k}=L^2(C_k,\cB(C_k),\tau_{C_k})=\ell^2(C_k) \,\,\,\,\,\text{and}\,\,\,\,\,\, \cH_{2,k}=\cK\oplus \bigoplus_{c\in C_k} \cK_{c,k}
		\end{equation*}
	where $\tau_{C_k}$ is the counting measure on $C_k$. We now define $\omega_{1,k}$ to be multiplication by the map $f_k(c)=(c+1)2^{-k}$ in $\cH_{1,k}$ and
	\begin{equation*}
	\omega_{2,k}=\omega \mid_{\cK}   \oplus \bigoplus_{c\in C_k} (c+1)2^{-k}.
	\end{equation*}
	Note $\omega_{1,k}\geq 2^{-k}$ and $\omega_{2,k}\geq 0$ since $C_k\subset \mathbb{N}_0$. Write $C_k=\{ n_{i,k} \}_{i=1}^K$ where $K\in \NN\cup \{ \infty \}$ and $ n_{i,k}< n_{i+1,k}$. Then  $\{1_{\{n_{i,k}\}}\}_{i=1}^K$ is an orthonormal basis of eigenvectors for $\omega_{1,k}$ corresponding to the eigenvalues $\{ (n_{i,k} +1)2^{-k}\}_{i=1}^K$. This collection of eigenvalues is either finite or diverges to infinity so $\omega_{1,k}$ will have compact resolvents. For each $g \in (0,\infty)$ and $c\in C_k$ we define the vector
	\begin{align*}
	\psi_{c,g,k}=\begin{cases}
	\frac{P_{c,k} v_g}{\lVert P_{c,k} v_g\lVert }, & P_{c,k} v_g\neq 0\\
	\text{Some normalized element in $P_{c,k}\cH$} & \text{otherwise}
	\end{cases}
	\end{align*}
	and note $\{ \psi_{c,g,k}\mid c\in C_k\}$ is an orthonormal collection of states. We also define
	\begin{align*}
	\widetilde{\cH}_{c,g,k}=\bigg\{ \psi\in P_{c,k}\cH \bigg \lvert \psi \perp \psi_{c,g,k} \bigg\}
	\end{align*}
	and note $\{ \widetilde{\cH}_{c,g,k} \mid c\in C_k\}$ consists of orthogonal subspaces. We then define
	\begin{equation*}
		\cH_{1,g,k}=\overline{\text{Span}\{ \psi_{c,g,k}\mid c\in C_k\}} \,\,\,\,\,\text{and}\,\,\,\,\,\, \cH_{2,g,k}=  \bigoplus_{c\in C_k} \widetilde{\cH}_{c,g,k}.
	\end{equation*}
	Using that $\omega$ is non negative and injective we find
	\begin{equation*}
	I=P_\omega((0,\widetilde{m}])+\sum_{c=0}^\infty P_{c,k}=P_\omega([0,\widetilde{m}])+\sum_{c\in C_k} P_{c,k},
	\end{equation*}
	which implies $\cH=\cH_{1,g,k}\oplus \cK\oplus  \cH_{2,g,k}$. Note that $v_g\in \cH_{1,g,k}$ by construction. Let $\mathcal{B}_{c,g,k}$ be an orthonormal basis for $\widetilde{\cH}_{c,g,k}$ and let $\mathcal{B}_{g,k}=\cup_{c\in C_k}\cB_{c,g,k}$ which is an orthonormal basis for $\cH_{2,g,k}$. Let $B\subset \cK$ be an orthonormal basis for $\cK$ and define $B_{g,k}=\{ \psi_{c,g,k}\mid c\in C_k\}$ which is an orthonormal basis for $\cH_{1,g,k}$. Define $D= \mathcal{B}_{g,k}\cup B_{g,k}\cup B$ which is an orthonormal basis for $\cH$.

	Let $V_{c,g,k}$ be a unitary from $\widetilde{\cH}_{c,g,k}$ to $\cK_{c,k}$ which exists since the spaces have the same dimension. Define $Q_{g,k}:\cH_{1,g,k}\rightarrow \cH_{1,k}$ to be the unique unitary map which satisfies $Q_{g,k}\psi_{c,g,k}=1_{\{ c \}}$. Then we define
	\begin{equation*}
	U_{g,k}=Q_{g,k}\oplus 1\oplus \bigoplus_{c\in C_k} V_{c,g,k}:\cH\rightarrow \cH_{1,k}\oplus \cH_{2,k}.
	\end{equation*}
	We now prove that
	\begin{equation}\label{eq:Operatoreq}
	U_{g,k}^*\omega_{1,k}\oplus \omega_{2,k}U_{g,k}=\omega_k.
	\end{equation}
	Let $\psi\in \cB_{c,g,k}\cup \{ \psi_{c,g,k} \}$ for some $c\in C_k$. Using the functional calculus we find $\psi=P_{c,k}\psi\in \cD(\omega_k)$ and 
	\begin{equation}\label{eigenvalue}
	\omega_k\psi=\omega_kP_{c,k}\psi=(c+1)2^{-k}P\psi_{c,g,k}=(c+1)2^{-k}\psi.
	\end{equation}	
	Furthermore, for $\psi\in B\subset  \cK$ we find that $\psi\in \cD(\omega_k^p)$ for all $p\in \NN$ and we have the inequality $\lVert \omega_k^p\psi \lVert\leq \widetilde{m}^p\lVert \psi\lVert $. In particular $D$ is an orthonormal basis for $\cH$ consisting of analytic vectors for $\omega_k$ so $D$ spans a core for $\omega_k$. Hence it is enough to prove equation (\ref{eq:Operatoreq}) on $D$.
	
	Let $\psi\in \mathcal{B}_{g,k}\cup B_{g,k}$ and pick $c\in C_k$ such that $\psi\in \cB_{c,g,k}\cup \{ \psi_{c,g,k} \}$. If $\psi= \psi_{c,g,k}$ then $U_{g,k}\psi=(1_{\{c\}},0)$. Now $1_{\{c\}}\in \cD(\omega_{1,k})$ with $\omega_{1,k}1_{\{c\}}=(c+1)2^{-k}1_{\{c\}}$ so $U_{g,k}\psi=(1_{\{c\}},0)\in \cD(\omega_{1,k}\oplus \omega_{2,k})$ and
	\begin{equation*}
	U_{g,k}^*\omega_{1,k}\oplus \omega_{2,k}U_{g,k}\psi=(c+1)2^{-k} U_{g,k}^*(1_{\{c\}},0)=(c+1)2^{-k}\psi=\omega_k\psi
	\end{equation*}
	by equation (\ref{eigenvalue}). If $\psi \in \cB_{c,g,k}$ then $U_{g,k}\psi=(0,V_{c,g,k}\psi)$. By definition we have $V_{c,g,k}\psi\in \cK_{c,k}\subset \cD(\omega_2)$ with $\omega_{2,k}V_{c,g,k}\psi=(c+1)2^{-k}V_{c,g,k}\psi$. Hence  $U_{g,k}\psi=(0,V_{c,g,k}\psi)\in \cD(\omega_{1,k}\oplus \omega_{2,k})$ and
		\begin{equation*}
		U_{g,k}^*\omega_{1,k}\oplus \omega_{2,k}U_{g,k}\psi=(c+1)2^{-k} U_{g,k}^*(1_{\{c\}},0)=(c+1)2^{-k}\psi=\omega_k\psi
		\end{equation*}
	by equation (\ref{eigenvalue}).	If $\psi\in B\subset \cK$ we have $\psi\in \cD(\omega_k)\cap  \cD(\omega)$ and $\omega_k\psi=\omega \psi\in \cK$. In particular $U_{g,k}\psi=(0,\psi)\in \cD(\omega_{1,k} \oplus \omega_{2,k})$ and $U_{g,k} \omega_k\psi= (0, \omega_k\psi)=(0, \omega\psi)$. Thus we find
	 \begin{equation*}
	 	U_{g,k}^*\omega_{1,k}\oplus \omega_{2,k}U_{g,k}\psi=U_{g,k}^*(0,\omega \psi)=	\omega_k\psi.
	 	\end{equation*}
	This proves equation (\ref{eq:Operatoreq}). As earlier noted $v_g\in \cH_{1,g,k}$ so $U_{g,k}v_g$ is of the form $(\widetilde{v}_{g,k},0)$ with $\lVert \widetilde{v}_{g,k}\lVert=\lVert v_g\lVert$. Using Lemma \ref{Lem:SeconduantisedBetweenSPaces} we find
	\begin{equation*}
	\Gamma(U_{g,k}) \widetilde{F}_{\eta}(v_g,\omega_k) \Gamma(U_{g,k})^*=\widetilde{F}_{\eta}( (\widetilde{v}_{g,k},0),\omega_{k,1}\oplus \omega_{k,2}  ).
	\end{equation*}
	Letting $L_1$ be the isomorphism from Lemma \ref{Iso1} we see that
	\begin{align*}
	L_1 \widetilde{F}_{\eta}( (\widetilde{v}_{g,k},0),\omega_{k,1}\oplus \omega_{k,2}  ) L_1^*=& d\Gamma(\omega_1)\otimes 1+1\otimes d\Gamma(\omega_{k_2})\\&+ \eta W(\widetilde{v}_{g,k},-1)\otimes \Gamma(-1).
	\end{align*}
	Letting $L_2$ be the isomorphism from Lemma \ref{Iso2} we see that
	\begin{align*}
	L_2L_1& \widetilde{F}_{\eta}((\widetilde{v}_{g,k},0),\omega_{k,1}\oplus \omega_{k,2}   ) L_1^*L_2^*\\&= d\Gamma(\omega_{k,1})+ d\Gamma^{(0)}(\omega_{k,2})+ \Gamma^{(0)}(-1)\mid_{\CC}\eta W(\widetilde{v}_{g,k},-1) \oplus \\&  \,\,\,\,\,\,\,\, \bigoplus_{n=1}^\infty d\Gamma(\omega_{k,1})\otimes 1+1\otimes d\Gamma^{(n)}(\omega_{k,2})+ \eta W(\widetilde{v}_{g,k},-1)\otimes \Gamma^{(n)}(-1)  
	\\&=\widetilde{F}_{\eta}(\widetilde{v}_{g,k},\omega_{k,1})\oplus \bigoplus_{n=1}^\infty \widetilde{F}_{(-1)^n\eta}(\widetilde{v}_{g,k},\omega_{k,1})\otimes1+1\otimes d\Gamma^{(n)}(\omega_{k,2})
	\end{align*}
	where we used $\Gamma^{(n)}(-1)=(-1)^n$. Hence $\cU_{g,k}=L_2L_1\Gamma(U_{g,k})$ will work.\cqfd
\end{proof}
\begin{lemma}\label{Compactcase}
	Assume that $\omega$ is selfadjoint, injective and non negative operator on $\cH$ which have compact resolvents. Let $\{ v_g \}_{g\in (0,\infty)}$ be a collection of elements in $\cH$ such that $\lVert v_g \lVert $ diverges to $\infty$. Then $\widetilde{F}_\eta(v_g,\omega)$ converges in norm resolvent sense to $\widetilde{F}_0(0,\omega)=d\Gamma(\omega)$ as $g$ goes to $\infty$ for all $\eta\in \RR$. 
\end{lemma}
\begin{proof}
	We calculate
	\begin{align*}
	(\widetilde{F}_{\eta}&(v_g,\omega)-i)^{-1}-(d\Gamma(\omega)-i)^{-1}\\=&\eta (\widetilde{F}_{\eta}(v_g,\omega)-i)^{-1}W(v_g,-1)(d\Gamma(\omega)-i)^{-1}\\=& \eta^2(\widetilde{F}_{\eta}(v_g,\omega)-i)^{-1}W(v_g,-1)(d\Gamma(\omega)-i)^{-1}W(v_g,-1)(d\Gamma(\omega)-i)^{-1}\\&+\eta(d\Gamma(\omega)-i)^{-1}W(v_g,-1)(d\Gamma(\omega)-i)^{-1}.
	\end{align*}
    This implies 
	\begin{align*}
	\lVert (\widetilde{F}_{\eta}(v_g,\omega)-i)^{-1}&-(d\Gamma(\omega)-i)^{-1}\lVert\\&\leq (\lvert \eta\lvert+1)\lvert\eta\lvert\lVert (d\Gamma(\omega)-i)^{-1}W(v_g,-1)(d\Gamma(\omega)-i)^{-1}\lVert,
	\end{align*}
which converges to 0 by Lemma \ref{WeakConv} and compactness of $(d\Gamma(\omega)-i)^{-1}$.\cqfd
\end{proof}

\begin{lemma}\label{Expkonv}
	Let $\cH$ be a Hilbert space. Let $\{A_n\}_{n=1}^\infty $ be a sequence of selfadjoint operators on $\cH$ that are uniformly bounded below by $\gamma$. Let $A$ be selfadjoint on $\cH$ and bounded below. Then $A_n$ converges to $A$ in norm resolvent sense if and only if $e^{-tA_n}$ converges to $e^{-tA}$ in norm for all $t<0$. In this case $\inf(\sigma(A_n))$ converges to $\inf(\sigma(A))$.
\end{lemma}
\begin{proof}
Norm resolvent convergence along with existence the uniform lower bound implies convergence of the semigroup (see \cite[Theorem VIII.20]{RS1}). To prove the converse we apply the formula
\begin{equation*}
(A-\lambda)^{-1}\psi=\int_{0}^\infty e^{-t(A-\lambda)}\psi dt
\end{equation*}
for all $\lambda<\gamma$ along with dominated convergence. To prove the last part note that by the spectral theorem
\begin{align*}
\inf(\sigma(A))=-\log(  \lVert  \exp(-A) \lVert  )=\lim_{n\rightarrow \infty} -\log(  \lVert  \exp(-A_n) \lVert  )=\lim_{n\rightarrow \infty}\inf(\sigma(A_n))
\end{align*}
finishing the proof.	\cqfd
\end{proof}

\begin{lemma}\label{ConvSimpleCase}
Assume that $\omega$ is a selfadjoint, injective and non negative operator on $\cH$. Let $\{ v_g \}_{g\in (0,\infty)}$ be a collection of elements in $\cH$ such that $\lVert v_g \lVert $ diverges to $\infty$. Assume there is $\widetilde{m}>0$ such that $P_{\omega}((0,\widetilde{m}])v_g=0$ for all $g$ where $P_\omega$ is the spectral measure corresponding to $\omega$. Then $\widetilde{F}_\eta(v_g,\omega)$ converges in norm resolvent sense to $\widetilde{F}_0(0,\omega)=d\Gamma(\omega)$ as $g$ tends to $\infty$ for all $\eta\in \RR$.
\end{lemma}
\begin{proof}
For each $k\in \mathbb{N}$ let $\cH_{1,k},\cH_{2,k},\omega_k,\omega_{1,k},\omega_{2,k}$ and $\widetilde{v}_{g,k}$ be the quantities from Lemma \ref{MainTech} corresponding to the family $\{ v_{g} \}_{g\in (0,\infty)}$ and the number $\widetilde{m}>0$. For each $n\in \NN_0$ we define
	\begin{equation*}
	F_{\pm \eta,k,g,n}=\widetilde{F}_{ \pm \eta}(\widetilde{v}_{g,k},\omega_{k,1})\otimes 1+1\otimes d\Gamma^{(n)}(\omega_{k,2}).
	\end{equation*}
	By Lemma \ref{MainTech} statement (1), it is enough to prove that $ \widetilde{F}_{ \eta}(v_{g},\omega_k)$ converges to $ d\Gamma(\omega_k)$ in norm resolvent sense as $g$ tends to $\infty$. Noting that $\widetilde{F}_{\eta}(v_{g},\omega_k)\geq -\lvert \eta\lvert$ for all $g$ we may use Lemma \ref{Expkonv}. Using the unitary transformations in Lemma \ref{MainTech} we see
	\begin{align*}
	&\lVert \exp(-t\widetilde{F}_{ \eta}(v_{g},\omega_k))-\exp(-td\Gamma(\omega_k)) \lVert \\&=  \sup _{n\in \mathbb{N}_0}\{ \lVert \exp(-tF_{ (-1)^n\eta,k,g,n})-\exp(-tF_{ 0,k,g,n})  \lVert   \}\\&=\sup _{n\in \mathbb{N}_0}\{ \lVert \exp(-t\widetilde{F}_{ (-1)^n \eta }(\widetilde{v}_{g,k},\omega_{k,1}))-\exp(-td\Gamma(\omega_{1,k}))  \lVert \lVert \exp (-td\Gamma^{(n)}(\omega_{2,k}))\lVert    \}\\&\leq \sup _{n\in \mathbb{N}_0}\{ \lVert \exp(-t\widetilde{F}_{ (-1)^n \eta }(\widetilde{v}_{g,k},\omega_{k,1}))-\exp(-td\Gamma(\omega_{1,k}))  \lVert    \}\\&=\max_{n\in \{1,2\}} \{ \lVert \exp(-t\widetilde{F}_{ (-1)^n \eta}(\widetilde{v}_{g,k},\omega_{k,1}))-\exp(-td\Gamma(\omega_{1,k}))  \lVert  \}
	\end{align*}
	which converges to 0 by Lemma \ref{Compactcase}. This finishes the proof.\cqfd
\end{proof}

\begin{lemma}\label{LemTran}
Assume that $\omega$ is selfadjoint, injective and non negative operator on $\cH$. Let $v\in \cD(\omega^{-1/2})$ and $f\in \cD(\omega)$. Then
\begin{align*}
W(f,1)F_{\eta}(v,\omega)W(f,1)^*=&\eta W(2f,-1)+d\Gamma(\omega)+\varphi(v-\omega f)\\&+\lVert \omega^{1/2} f\lVert^2-2\textup{Re}(\langle v ,f \rangle).
\end{align*}
\end{lemma}
\begin{proof}
Use equation (\ref{eq:Weylcomp}), (\ref{grupper1}) and Lemma \ref{Weyltrans}. \cqfd
\end{proof}
\noindent We can now prove Theorem \ref{Centralconv}.
\begin{proof}[of Theorem \ref{Centralconv}] The formula in equation (\ref{eq:LIGHED}) is obtained via Lemma \ref{LemTran} and the lower bound is trivial from Lemma \ref{Lem:FundamentalIneq}. For $c\in (0,\infty)$ we will write $P_c=P_\omega((c,\infty))$ and $\overline{P}_c=1-P_c=P_\omega((0,c])$. Note that $P_cv_g,P_cv,\overline{P}_cv_g,\overline{ P}_cv\in \cD(\omega^{-1/2})$ holds trivially by the spectral theorem. Define for $ 0<c\leq \widetilde{m}$ and $\eta\in \RR$ 
\begin{align*}
\widetilde{F}_{\eta,c,g}&:=\eta W(2\omega^{-1}P_{c}v_g,-1)+d\Gamma(\omega)+\varphi(\overline{P}_c v_g)+\lVert \omega^{-1/2}P_{\omega}((c,\widetilde{m}])v_g\lVert^2\\
A_{\eta,c,g}&:=\eta W(2\omega^{-1}P_{c}v_g,-1)+d\Gamma(\omega)+\lVert \omega^{-1/2}P_{\omega}((c,\widetilde{m}])v_g\lVert^2
\end{align*}
and note they are all selfadjoint on $\cD(d\Gamma(\omega))$ by the Kato-Rellich theorem and Lemma \ref{Lem:FundamentalIneq}. For $\psi\in \cD(d\Gamma(\omega))$ we have  $\lVert (1+d\Gamma(\omega))^{1/2}\psi\lVert\leq \lVert (1+d\Gamma(\omega))\psi\lVert$ by the spectral theorem. Using this and Lemma \ref{Lem:FundamentalIneq} we find for all $c\in (0,\widetilde{m}]$
\begin{align*}
\lVert (\widetilde{F}_{\eta,c,g}+i)^{-1}&-(\widetilde{A}_{\eta,c,g}+i)^{-1}\lVert \\&\leq \lVert \varphi(\overline{P}_c v_g)(\widetilde{A}_{\eta,c,g}+i)^{-1}\lVert\\& \leq 2\lVert \overline{P}_c(1+\omega^{-1/2})v_g \lVert \lVert (1+d\Gamma(\omega)) (\widetilde{A}_{\eta,c,g}+i)^{-1}\lVert\\&\leq 2\lVert \overline{P}_c(1+\omega^{-1/2})v_g \lVert(1+ 1+1+\lvert \eta\lvert+\lVert \omega^{-1/2}P_{\omega}((0,\widetilde{m}])v_g\lVert^2 ).
\end{align*}
where we in the last step used $\lVert \omega^{-1/2}P_{\omega}((c,\widetilde{m}])v_g\lVert^2\leq \lVert \omega^{-1/2}P_{\omega}((0,\widetilde{m}])v_g\lVert^2$. We now define 
\begin{align*}
C_1:=3+\lvert  \eta \lvert +\sup_{g\in (0,\infty)}\lVert \omega^{-1/2}P_{\omega}((0,\widetilde{m}])v_g\lVert^2
\end{align*}
which is finite since $\omega^{-1/2}P_{\omega}((0,\widetilde{m}])v_g$ is convergent. Let $B=d\Gamma(\omega)+\varphi(v)$ and $C_2=\lVert (1+d\Gamma(\omega)) (B+i)^{-1}\lVert$. We estimate using Lemma \ref{Lem:FundamentalIneq}
\begin{align*}
\lVert (\widetilde{F}_{0,\widetilde{m},g}+i)^{-1}-(B+i)^{-1}\lVert&\leq \lVert \varphi(\overline{P}_{\widetilde{m}} v_g-v)(B+i)^{-1}\lVert \\& \leq 2\lVert (1+\omega^{-1/2})(\overline{P}_{\widetilde{m}}v_g-v) \lVert C_2
\end{align*}
Let $U_c=W(P_{\omega}((c_1,\widetilde{m}])v,1)$ for some $c\in (0,\widetilde{m}]$. Using equation (\ref{eq:Weylcomp}) and Lemma \ref{Weyltrans} we obtain $U_c\widetilde{F}_{\eta,\widetilde{m},g}U_c^*=\widetilde{F}_{\eta,c,g}$ for all $\eta\in \RR$. Using this transformation and the previous estimates we find for all $c\in (0,\widetilde{m}]$ and $g>0$ that
 \begin{align*}
\lVert (\widetilde{F}_{\eta,\widetilde{m},g}+i)^{-1}-(B+i)^{-1} \lVert\leq& \lVert (\widetilde{F}_{\eta,c,g}+i)^{-1}-(\widetilde{F}_{0,c,g}+i)^{-1} \lVert \\&+ \lVert (\widetilde{F}_{0,\widetilde{m},g}+i)^{-1}-(B+i)^{-1}\lVert \\\leq& 2(2C_1\lVert \overline{P}_c(1+\omega^{-1/2})v_g \lVert)\\&+ \lVert (\widetilde{A}_{\eta,c,g}+i)^{-1}-(\widetilde{A}_{0,c,g}+i)^{-1} \lVert   \\&+2C_2\lVert (1+\omega^{-1/2})(\overline{P}_{\widetilde{m}}v_g-v) \lVert
 \end{align*}
Noting that
\begin{align*}
\lVert \overline{P}_c(1+\omega^{-1/2})v_g \lVert\leq \lVert (1+\omega^{-1/2})(\overline{P}_{\widetilde{m}}v_g-v) \lVert+\lVert \overline{ P}_c (1+\omega^{-1/2})v \lVert
\end{align*}
we see that
\begin{align*}
\limsup_{g\rightarrow \infty}  \lVert (\widetilde{F}_{\eta,\widetilde{m},g}+i)^{-1}-(B+i)^{-1} \lVert\leq 4C_1\lVert \overline{ P}_c (1+\omega^{-1/2})v \lVert
\end{align*}
for all $c\in (0,\widetilde{m}]$ by Lemma \ref{ConvSimpleCase}. Taking $c$ to 0 finishes the proof since $\widetilde{F}_{\eta,\widetilde{m},g}=\widetilde{F}_{\eta,\widetilde{m}}(v_g,\omega)$.
 \cqfd
\end{proof}

\begin{proposition}\label{Counter}
	Let $\cH=L(\mathbb{R}^3,\cB(\RR^3),\lambda^{\otimes 3}),\omega(k)=\lvert k\lvert $, $v_g=\omega^{-1/2}1_{\{g^{-1} \leq \lvert k\lvert\leq 2\}}$ and $\eta<0$. Then $\lVert \omega^{-1}v_g \lVert $ converges to $\infty$ for $g$ going to $\infty$, but there is $h\neq 0$ such that $\widetilde{F}_\eta(hv_g,\omega)=d\Gamma(\omega)+\eta W(h\omega^{-1}v_g,-1)$ does not go to $d\Gamma(\omega)$ in norm resolvent sense.
\end{proposition}
\begin{proof}
	Define $v=\omega^{-1/2}1_{ \{\lvert k\lvert\leq 2\}}$. It is easy to see that $\lVert \omega^{-1}v_g \lVert $ goes to $\infty$ as $g$ tends to infinity. Assume that convergence in norm resolvent sense holds for all $h\neq 0$. Applying Lemma \ref{LemTran} with $f=h\omega^{-1}v_g$ we see
	\begin{equation*}
	\inf\{ \sigma(F_{\eta}(hv_g,\omega)) \} +h^2 \lVert \omega^{-1/2}v_g\lVert^2=\inf\{ \sigma(\widetilde{F}_{\eta}(hv_g,\omega)) \}
	\end{equation*}
	converges to $0$ for $g$ going to $\infty$. In \cite{Thomas1} it is proven that $\inf\{ \sigma(F_{\eta}(hv_g,\omega)) \}$ converges in norm resolvent sense to $\inf\{ \sigma(F_{\eta}(hv,\omega)) \}$. Hence we find
	\begin{equation*}
	\inf\{ \sigma(F_{\eta}(hv,\omega)) \} =-h^2 \lVert \omega^{-1/2}v\lVert^2
	\end{equation*}
	for any $h\neq 0$. Since $\langle \Omega,F_{\eta}(hv,\omega)\Omega  \rangle=\eta$ we see $-h^2 \lVert \omega^{-1/2}v\lVert^2\leq \eta<0$ for all $h\neq 0$. Taking $h$ to 0 yields $0\leq \eta<0$ which is a contradiction. \cqfd
\end{proof}

\noindent We now prove Corollaries \ref{KorGSFiber}, \ref{KorGSTotal} and \ref{KorUVREN}.
\begin{proof}[of Corollary \ref{KorGSFiber}]
	Define $U_g=W(g\omega^{-1}v,1)$ and $m:=m(\omega)>0$. Note that
	\begin{equation*}
	\widetilde{F}_\eta(2g\omega^{-1}v,\omega):=U_g F_{\eta}(gv,\omega)U_g^*+g^2\lVert \omega^{-1/2}v\lVert^2
	\end{equation*}
    converges in norm resolvent sense to $d\Gamma(\omega)$ as $g$ tends to infinity by Theorem \ref{Centralconv} (use $\widetilde{m}=m(\omega)$ and $v=0$). Since $\widetilde{F}_\eta(2g\omega^{-1}v,\omega)\geq -\lvert \eta \lvert$ for all $g>0$, Lemma \ref{Expkonv} implies
	\begin{equation*}
	\lim_{g\rightarrow \infty}  \cE_{\eta}(gv,\omega) +g^2\lVert \omega^{-1/2}v \lVert^2 =\lim_{g\rightarrow \infty}  \inf(\sigma(\widetilde{F}_\eta(2g\omega^{-1}v,\omega)))=0.
	\end{equation*}
	Let $P_g$ be the spectral projection of  $\widetilde{F}_\eta(gv,\omega)$ onto  $[-\frac{m}{2},\frac{m}{2}]$. Using \cite[Theorem VIII.23]{RS1} and Lemma \ref{Lem:SecondQuantisedProp} we find $P_g$ converges in norm to $P=\lvert \Omega \rangle \langle \Omega \lvert $. Pick $g_0$ such that $\cE_{\eta}(gv,\omega)+g^2\lVert \omega^{-1/2}v\lVert^2\in (-\frac{m}{2},\frac{m}{2})$ and $\lVert P_g-P \lVert<1 $  for all $g>g_0$. Then $P_g$ has dimension 1 by \cite[Theorem 4.35]{Weidmann}, so $\widetilde{F}_\eta(gv,\omega)$ and $F_{\eta}(gv,\omega)$ will have a non degenerate isolated ground state for all $g>g_0$. Let $\{ \psi_g \}_{g\geq g_0}$ be a real analytic collection of normalized eigenstates for $F_{\eta}(gv,\omega)$ and write $\widetilde{\psi}_g=U_g\psi_g$ which is a ground state for $\widetilde{F}_\eta(gv,\omega)$. We calculate
	\begin{equation*}
	\lvert \langle  e^{-g^2\lVert \omega^{-1}v\lVert} \epsilon(-g\omega^{-1}v) ,\psi_g\rangle\lvert =\lvert \langle \Omega,\widetilde{\psi}_g \rangle\lvert =\lVert P\widetilde{\psi}_g \lVert\geq \lVert \widetilde{\psi}_g \lVert-\lVert P-P_g\lVert \lVert \widetilde{\psi}_g \lVert>0.
	\end{equation*}
	Hence $h(g)= \langle  e^{-g^2\lVert \omega^{-1}v\lVert} \epsilon(-g\omega^{-1}v),\psi_g \rangle$ is nonzero and smooth. Multiplying with $\frac{1}{ h(g)}$ and normalising, we may pick the family $\{ \psi_g \}_{g\geq g_0}$ smooth such that 
	\begin{equation*}
	\langle  e^{-g^2\lVert \omega^{-1}v\lVert} \epsilon(-g\omega^{-1}v),\psi_g \rangle=\langle \Omega,U_g \psi_g \rangle> 0 .
	\end{equation*}
	This implies
	\begin{equation*}
	1\geq \lvert\langle \Omega, \widetilde{\psi}_g \rangle\lvert=\lVert P\widetilde{\psi}_g \lVert\geq \lVert \widetilde{\psi}_g \lVert-\lVert P-P_g\lVert \lVert \widetilde{\psi}_g \lVert=1-\lVert P-P_g\lVert. 
	\end{equation*}
	Therefore $\lvert \langle \Omega,\widetilde{\psi}_g \rangle\lvert=\langle \Omega,\widetilde{\psi}_g \rangle$ converges to 1, and hence $\widetilde{\psi}_g$ converges to $\Omega$. This implies
	\begin{equation*}
	0=\lim_{g\rightarrow \infty} \lVert \widetilde{\psi}_g -\Omega \lVert= \lim_{g\rightarrow \infty} \lVert U^*_g\widetilde{\psi}_g -U_g^*\Omega \lVert=\lim_{g\rightarrow \infty} \lVert \psi_g -e^{-g^2\lVert \omega^{-1}v \lVert^2}\epsilon(-g\omega^{-1}v) \lVert.
	\end{equation*}
	Using Lemma \ref{WeakConv} we find
	\begin{equation*}
	\langle \widetilde{\psi}_g,d\Gamma(\omega) \widetilde{\psi}_g \rangle=\cE_{\eta}(gv,\omega)+g^2\lVert \omega^{-1/2}v\lVert^2-\eta \langle \widetilde{\psi}_g,W(2g\omega^{-1}v,-1)\widetilde{\psi}_g \rangle
	\end{equation*}
	 converges to 0. Hence $\psi_g$ converges to $\Omega$ in $d\Gamma(\omega)^{1/2}$ norm, and hence also in $N^{1/2}$ norm since $m>0$. Note that $\psi_g,\widetilde{\psi}_g\in \cD(d\Gamma(\omega))\subset \cD(N)$ since $m>0$. Using Theorem \ref{Weyltrans} we see that
	\begin{equation*}
	\langle \psi_g,N\psi_g \rangle=\langle \widetilde{\psi}_g,U_gNU^*_g\widetilde{\psi}_g \rangle= \langle \widetilde{\psi}_g,N\widetilde{\psi}_g \rangle+g\langle \widetilde{\psi}_g,\varphi(\omega^{-1}v)\widetilde{\psi} \rangle +g^2\lVert \omega^{-1}v\lVert.
	\end{equation*}
	Since $\widetilde{\psi}_g$ goes to $\Omega$ in $N^{1/2}$ norm and $\varphi(\omega^{-1}v)$ is $N^{1/2}$ bounded by Lemma \ref{Lem:FundamentalIneq} we find that $\varphi(\omega^{-1}v)\widetilde{\psi}_g$ converges to $\varphi(\omega^{-1}v)\Omega$ in norm. Hence $\langle \widetilde{\psi}_g,\varphi(\omega^{-1}v)\widetilde{\psi}_g \rangle$ and $\langle \widetilde{\psi}_g,N\widetilde{\psi}_g \rangle$ converges to 0 which implies
	\begin{equation*}
	(g^{-1}\langle \psi_g,N\psi_g \rangle-g\lVert \omega^{-1}v\lVert )= g^{-1}\langle \widetilde{\psi}_g,N\widetilde{\psi}_g \rangle+\langle \widetilde{\psi}_g,\varphi(\omega^{-1}v)\widetilde{\psi} \rangle
	\end{equation*}
	converges to 0 as $g$ tends to $\infty$. Define $f(g)=\cE_{\eta}(gv,\omega)+g^2\lVert \omega^{-1/2}v\lVert^2$ and assume  $\eta < 0$. Since $f(0)=\eta$ and $f$ converges to 0, we just need to see $f$ is increasing. There is a unitary map $U:\cH\rightarrow L^2(X,\cF,\mu)$ such that $U\omega U^*$ is a multiplication operator.
	Using Lemma \ref{Lem:SeconduantisedBetweenSPaces} we see $\Gamma(U)F_\eta(gv,\omega)\Gamma(U)^*=F_\eta(gUv,U\omega U^*)$ so 
	\begin{align*}
	f(g)=\cE_{\eta}(gUv,U\omega U^*)+g^2\lVert (U\omega U^*)^{-1/2}Uv\lVert^2
	\end{align*}
	Hence we may assume $\cH=L^2(\cM,\cF,\mu)$ and $\omega$ is multiplication by a strictly positive and measurable map, which we shall also denote $\omega$. We note $\psi_g$ exists for all $g\geq 0$ by Theorem \ref{Thm:Spectral Theory of decomposition} and we have the pull through formula (see \cite{Thomas1})
\begin{equation}\label{eq:pull1}
a(k)\psi_g=-gv(k)(F_{-\eta}(gv,\omega)-\cE_\eta(gv,\omega)+\omega(k))^{-1}\psi_g.
\end{equation}	
 Note that $g\mapsto \cE_{\eta}(gv,\omega)$ is real analytic since it is a an isolated non degenerate eigenvalue by Theorem \ref{Thm:Spectral Theory of decomposition}.  We may then calculate
	\begin{align*}
	\frac{d}{dg} \cE_{\eta}(gv,\omega)&=\langle \psi_g,\varphi(v)\psi_g \rangle=2 \text{Re}(\langle \psi_g,a(v)\psi_g \rangle)\\&=-2g  \int_{X} \lvert v(k)\lvert^2 \lVert  (F_{-\eta}(gv,\omega)-\cE_{\eta}(gv,\omega)+\omega(k))^{-1/2} \psi_g\lVert^2dk\\&>-2g\lVert \omega^{-1/2}v\lVert^2=-	\frac{d}{dg} g^2\lVert \omega^{-1/2}v\lVert
	\end{align*}
	because $\lVert  (F_{-\eta }(gv,\omega)-\cE_{\eta}(gv,\omega)+\omega(k))^{-1/2} \psi_g\lVert^2<\omega(k)^{-1}$ almost everywhere by Theorem \ref{Thm:Spectral Theory of decomposition}. This proves the claim. \cqfd
\end{proof}

\begin{proof}[of Corollary \ref{KorGSTotal}]
	By Theorem \ref{Thm:Spectral Theory of decomposition} we may pick a a unitary map $V$ such that $VH_{\eta}(gv,\omega)V^*=F_{-\eta }(gv,\omega)\oplus F_{\eta}(gv,\omega)$. Noting that
	\begin{equation*}
	\lim\limits_{g\rightarrow \infty} \cE_{\eta}(gv,\omega)- \cE_{-\eta }(gv,\omega)=0
	\end{equation*}
	and that $\cE_{\pm \eta }(gv,\omega)$ is an eigenvalue for $F_{\pm \eta}(gv,\omega)$ for sufficiently large $g$ by Corollary \ref{KorGSFiber}, we see that for $g$ large enough $H_{\eta}(gv,\omega)$ will have at least two eigenvalues in the mass gap $[E_{\eta}(gv,\omega),E_{\eta}(gv,\omega)+m_{\textup{ess}}]$, and the energy difference will converge to 0. \cqfd
\end{proof} 

\begin{proof}[of Corollary \ref{KorUVREN}]
	 Define $U_g=W(\omega^{-1}v_g,1)$ which is independent of $\eta $. If $\omega^{-1}v1_{\{ \omega\geq 1  \}} \notin \cH$, we see part (1) follows from Theorem \ref{Centralconv} and Lemma \ref{Lem:FundamentalIneq}.
	
	We now prove part (2). By Theorem \ref{Thm:Spectral Theory of decomposition} there is a unitary map $V$  with the property that $VH_{\eta}(v_g,\omega)V^*=F_{\eta }(v_g,\omega)\oplus F_{-\eta}(v_g,\omega)$. Let $V_g=V^* (U_g\oplus U_g) V$. Convergence to $\widetilde{H}$ and the uniform lower bound now follows from part (1) and Lemma \ref{Lem:FundamentalIneq}. Let $C_g=\lVert \omega^{-1/2}1_{\{\omega>\widetilde{m} \}}v_g\lVert^2$. Then
	\begin{align*}
	\lVert &(H_{\eta}(v_g,\omega)+C_g+i)^{-1}- (H_{0}(v_g,\omega)+C_g+i)^{-1}\lVert\\&=\lVert (V_gH_{\eta}(v_g,\omega)V_g^*+C_g+i)^{-1} -(V_gH_{0}(v_g,\omega)V_g^*+C_g+i)^{-1}\lVert
	\end{align*}
	which converges to 0. This finishes the proof. \cqfd
\end{proof}

\section{Uniqueness, support, pointwise bounds and number bounds}
In this chapter we prove Theorem \ref{BasicpropGS}. We will in this section assume that $\cH=L^2(X,\cF,\mu)$ where $(X,\cF,\mu)$ is $\sigma$-finite and countably generated. We will also assume $\omega$ is a multiplication operator which satisfies  $\omega>0$ almost everywhere. We also fix $v\in \cD(\omega^{-1/2})$ and define
\begin{align}\label{eq:denh}
h(x)=\begin{cases} -1 & v(x)=0 \\
-\frac{v(x)}{\lvert v(x)\lvert} & v(x)\neq 0
\end{cases}
\end{align}
Note that $h$ is measurable, $\lvert h \lvert=1$ and $h^*v=-\lvert v\lvert$. Define also $h^{(n)}(k_1,\dots,k_n)=h(k_1)\dots h(k_n)$ and note
\begin{equation*}
\Gamma(h)=\bigoplus_{n=0}^\infty h^{(n)} \,\,\,\,\,\,\,\, \Gamma(h)^*=\Gamma(h^*)=\bigoplus_{n=0}^\infty (h^{(n)})^*.
\end{equation*}
Define
\begin{equation*}
\cC_+=\{ \psi=(\psi^{(n)}) \in \cF_b(\cH) \mid \psi^{(0)}\geq 0, (h^{(n)})^*\psi^{(n)}\geq 0\,\,\, \text{a. e. for $n\geq 1$}  \}.
\end{equation*}
We have
\begin{lemma}
	$\cC_+$ is a selfddual cone inside $\cF_b(\cH)$. The strictly positive elements are
	\begin{equation*}
	\cC_{>0}=\{ \psi=(\psi^{(n)}) \in \cF_b(\cH) \mid \psi^{(0)}>0, (h^{(n)})^*\psi^{(n)}> 0\,\,\, \text{a. e. for $n\geq 1$}  \}.
	\end{equation*}
\end{lemma}
\begin{proof}
	We note
	\begin{equation*}
	\cC_+=\Gamma(h)\{ \psi=(\psi^{(n)})_{n=0}^\infty \in \cF_b(\cH) \mid \psi^{(0)}\geq 0, \psi^{(n)}\geq 0\,\,\, \text{a. e. for $n\geq 1$}  \}.
	\end{equation*}
	The result now follows by the theory developed in \cite{Farris}.
\end{proof}
\begin{lemma}\label{POS}
Let $g>0$, $\gamma\in \RR$ and $T$ be a selfadjoint operator on $\cF_b(\cH)$ such that
\begin{equation*}
T=\bigoplus_{n=0}^\infty T^{(n)}.
\end{equation*}
Assume $T^{(n)} $ a multiplication operator and $T^{(n)}\geq \gamma$ for all $n\in \NN_0$. Assume $g\varphi(v)$ infinitesimally $T$-bounded. Define $H=T+g\varphi(v)$. Then $H$ is bounded below, selfadjoint and $(H-\lambda)^{-1} \cC_+\subset \cC_+$ for all $\lambda<\inf(\sigma(H))$. If $v\neq 0$ almost everywhere then $(H-\lambda)^{-1}\cC_+\backslash \{0\}\subset \cC_{>0}$ for all $\lambda<\inf(\sigma(H))$. So if $\inf(\sigma(H))$ is an eigenvalue then it is non degenerate and spanned by an element in $\cC_{>0}$.
\end{lemma}
\begin{proof}
The Kato Rellich theorem implies $H$ is selfadjoint and bounded below. For $\lambda <  -\gamma$ we note that $(T-\lambda)^{-1}$ acts on each particle sector as multiplication with a positive bounded map. Hence it will map $\cC_+$ into $\cC_+$. Assume now that $\psi=(\psi^{(n)}) \in \cC_+\cap \cD( d\Gamma(\omega) )$. Then we have almost everywhere that
	\begin{align*}
	&(h^{(n-1)})^*(-a(v)\psi^{(n)})(k_2,\dots,k_n)=\sqrt{n}\int_{\cM} \lvert v(k)\lvert ((h^{(n)})^*\psi^{(n)})(k,k_2,\dots,k_n)dk \\
	&(h^{(n+1)})^*(-a^{\dagger}(v)\psi^{(n)})(k_1,\dots,k_{n+1})\\&\,\,\,\,\,\,=\frac{1}{\sqrt{n+1}}\sum_{l=1}^{n+1} -h^*(k_l)v(k_l)( (h^{(n)})^*\psi^{n})(k_1,\dots,\hat{k}_l,\dots,k_{n+1}) \
	\end{align*}
	which implies $-g\varphi(v)\psi\in \cC_{+}$. In particular we obtain 
	\begin{align}
	&(-g\varphi(v)(T-\lambda)^{-1})^{n}\cC_{+}\subset \cC_+\nonumber 
	\\&\label{eq:typeterterm} (-1)^ng^n \prod_{k=1}^n a^{\sharp_k}(v)  (T-\lambda)^{-1}\cC_{+}\subset \cC_+.
	\end{align}
	where $\sharp_k$ can be either a $\dagger$ or nothing. For $\lambda\in \RR$ sufficiently negative we may expand 
	\begin{equation}\label{eq:abe}
		(H-\lambda)^{-1}=\sum_{n=0}^{\infty} (T-\lambda)^{-1} (-g\varphi(v)(T-\lambda)^{-1})^n.
		\end{equation}
	Since each term preserves the closed set $\cC_+$ we find $(H-\lambda)^{-1} \cC_+\subset \cC_+$ for $\lambda$ small enough. Assume now $v\neq 0$ almost everywhere. Let $I_n$ denote the integral over $\cM^n$ with respect to $\mu^{\otimes n}$. For $u\in S_n(\cH^{\otimes n})\backslash \{ 0 \}$ with $u\in \cC_+$ we have 
	\begin{align}
	&(-1)^n(ga(v)(T-\mu)^{-1})^{n}u\label{eq:Positiv} \\&=I_n\left((-1)^ng^n\overline{v(k_1)\dots v(k_n)}u(k_1,\dots,k_n)\prod_{\ell=1}^{n} ( T^{(\ell)}(k_1,\dots,k_\ell) -\lambda  )^{-1}\right)\nonumber
	\end{align}
	which is strictly positive. Let $u,w\in \cC_+\backslash \{0\}$. Pick $n_1$ such that $u^{(n_1)}\neq0$ and $n_2$ such that $w^{(n_2)}\neq 0$. Consider now the $n=n_1+n_2$ term in equation (\ref{eq:abe}). This term can again be written as a sum of terms of the form \eqref{eq:typeterterm} multiplied to the left by $(T-\lambda)^{-1}$. Since all terms are positivity preserving we find
	\begin{align*}
	\langle u, (H-\lambda )^{-1}w \rangle\nonumber  &\geq \langle (T-\lambda)^{-1}u,( -ga^{\dagger}(v)  (T-\lambda)^{-1})^{n_1}(-ga(v)  (T-\lambda)^{-1})^{n_2} w \rangle\nonumber \\&= \langle (T-\lambda)^{-1}( -ga(v)  (T-\lambda)^{-1})^{n_1}u,(-  ga(v)(T-\lambda)^{-1})^{n_2} w \rangle
	\end{align*}
	Since $u-u^{(n_1)}\in \cC_+$ and $w-w^{(n_2)}\in \cC_+$ we find the following lower bound: 
	\begin{equation*}
(T^{(0)}-\lambda)^{-1}(-ga(v)  (T-\lambda)^{-1})^{n_1}u^{(n_1)}(-ga(v)  (T-\lambda)^{-1})^{n_2}w^{(n_2)}
	\end{equation*}
	which is strictly positive by Equation (\ref{eq:Positiv}). Hence we have proven the lemma for $\lambda$ sufficiently negative. Now fix $\lambda $ such that the lemma is true. For any $\mu\in (\lambda,\inf(\sigma(H)))$ we can use standard theory of resolvents to write
	\begin{equation*}
	(H-\mu )^{-1}=\sum_{k=0}^{\infty}(\mu-\lambda)^n((H-\lambda )^{-1})^{n+1}
	\end{equation*}
	which is positivity preserving/improving since each term is.  \cqfd
\end{proof}
The following lemma can be found in \cite{Thomas1}.
\begin{lemma}\label{Lem:GSreduced}
	Define $A=\{ v\neq 0 \}$, $\mu_A(B)=\mu(A\cap B)$ and $\mu_{A^c}(B)=\mu(A^c\cap B)$. Let $\cH_1=L^2(X,\cF,\mu_A)$ and $\omega_1$ be multiplication with $\omega$ but on the space $\cH_1$. Assume that $\eta\leq 0$ and $g>0$. Then
	\begin{enumerate}
		\item $\cE_{\eta}(gv,\omega)=\cE_\eta(gv,\omega_1)$ and $\cE_{\eta}(gv,\omega)$ is an eigenvalue for $F_{\eta}(gv,\omega)$ if and only if $\cE_{\eta}(gv,\omega)$ is an eigenvalue for $F_{\eta}(gv,\omega_1)$. In this case the dimension of the eigenspace is 1.
		
		\item If $\psi=(\psi^{(n)})_{n=0}^\infty$ is a ground state for $F_{\eta}(gv,\omega_1)$, then $\psi=(1_{A^n}\psi^{(n)})_{n=0}^\infty$ is a ground state for $F_{\eta}(gv,\omega)$.
	\end{enumerate}
\end{lemma}

\noindent We can now finally prove Theorem \ref{BasicpropGS}.

\begin{proof}[Proof of Theorem \ref{BasicpropGS}]
	Statement (1) follows from Lemmas \ref{POS} and \ref{Lem:GSreduced} since $g\phi(v)=\phi(gv)$ and $h$ defined in equation (\ref{eq:denh}) does not depend on $g$ as long as $g>0$. To prove statement (2) we let $\psi$ be a ground state for $F_{\eta}(gv,\omega)$. Define for $\lambda>0$ and $\ell\in \NN$ the operator
	\begin{equation*}
	R_\ell(\lambda)=(F_{(-1)^\ell \eta}(gv,\omega)+\lambda-\cE_{\eta}(gv,\omega))^{-1}.
	\end{equation*}
	This makes sense since $\cE_{\eta}(gv,\omega)\leq \cE_{-\eta }(gv,\omega)$ by Proposition \ref{Lem:BasicPropertiesSBmodel}. Using the pull through formula found in \cite{Thomas1} we find
	\begin{equation*}
	a(k_1,\dots,k_n)\psi_{g,\eta}=\sum_{i=1}^{n} gv(k_i)R_n(\omega(k_1)+\cdots+\omega(k_n))a(k_1,\dots,\hat{k}_i,\dots,k_n)\psi_{g,\eta},
	\end{equation*}
	where $\widehat{k}_i$ means that the variable $k_i$ is omitted and $a(k_1,\dots,k_n)=a(k_1)\cdots a(k_n)$ (a strict definition of such an expression can be found in \cite{Thomas1}). We proceed by induction to show that $\lVert a(k_1,\dots,k_n)\psi_g\lVert\leq g^n\frac{\lvert v(k_1)\lvert\dots\lvert v(k_n)\lvert}{\omega(k_1)\dots\omega(k_n)}$. For $k=1$ this follows since $\cE_{\eta}(gv,\omega)\leq \cE_{-\eta}(gv,\omega)$ and so
	\begin{equation*}
	\lVert a(k)\psi_{g,\eta}\lVert =\left \lVert  \frac{gv(k)}{F_{-\eta}(gv,\omega)+\omega(k)+\cE_{\eta}(gv,\omega)}\psi_g\right\lVert\leq g \frac{\lvert v(k) \lvert }{\omega(k)}.
	\end{equation*}
	Using the induction hypothesis we may now compute
	\begin{align*}
	\lVert a(k_1,\dots,k_n)\psi_{g,\eta}\lVert &\leq \sum_{i=1}^{n} \frac{\omega(k_i) }{\omega(k_1)+\cdots+\omega(k_n)}  g^{n}\frac{\lvert v(k_1)\lvert\dots\lvert v(k_n)\lvert}{\omega(k_1)\dots\omega(k_n)}\\&=g^{n}\frac{\lvert v(k_1)\lvert\dots \lvert v(k_n)\lvert}{\omega(k_1)\dots\omega(k_n)}.
	\end{align*}
	Now $\sqrt{n!}\lvert \psi_{g,\eta}^{(n)}(k_1,\dots,k_n)\lvert \leq \lVert a(k_1,\dots,k_n)\psi_{g,\eta}\lVert $ and so the desired inequality follows.
	
	Statement (3): By Theorem \ref{Thm:Spectral Theory of decomposition} and $N=d\Gamma(1)$ we see that the conclusions about $\phi_{g,\eta}$ follows from those of $F_{\eta}(gv,\omega )$. It is easily seen that $\psi_{g,0}=e^{-2^{-1}g^2\lVert \omega^{-1}v\lVert}\epsilon(g\omega^{-1}v)$ and $\epsilon(g\omega^{-1}v)\in \cD(f(N))\iff  \alpha_{g,f,v,\omega}<\infty$. This proves the $"\Leftarrow"$ part. If $\alpha_{g,f,v,\omega}<\infty$ then we may use the point wise bounds to obtain
	\begin{equation*}
	\sum_{n=0}^{\infty} f(n)^2 \lVert \psi^{(n)}_{g,\eta} \lVert^2\leq \sum_{n=0}^{\infty} \frac{f(n)^2g^{2n} \lVert \omega^{-1}v \lVert^{2n}  }{n!}<\infty,
	\end{equation*}
	which proves the $"\Rightarrow"$. \cqfd
\end{proof}

\section{Convergence in the massless case}
In this section we will assume $\cH=L^2(\RR^\nu,\cB(\RR^\nu),\lambda^{\nu})$ and that $\omega$ is a selfadjoint, non negative and injective multiplication operator on this space with $m(\omega)=0$. Fix an element $v\in \cD(\omega^{-1})\backslash \{0\}$. In \cite{Thomas1} it is proven that if $\eta\leq 0$ then $F_\eta(gv,\omega)$ has a normalised ground state $\psi_g$ for any $g\in \RR$ and $\cE_{\eta}(gv,\omega)=\cE_{-\eta}(gv,\omega)$. Furthermore we will for $\eta,g\in \RR$ write $F_{\eta,g}:=F_{\eta}(gv,\omega)$ and $\cE_{\eta,g}:=\cE_{\eta}(gv,\omega)$.
\begin{lemma}\label{CentralResMassless}
	Assume $\eta\leq 0$. Define $U_g=W(g\omega^{-1}v,1)$ and $\widetilde{\psi}_g=U_g\psi_g$. Then
	\begin{equation*}
	0\leq \langle \widetilde{\psi}_g,d\Gamma(\omega)\widetilde{\psi}_g \rangle  \leq \lvert \eta\lvert \langle \psi_g,\Gamma(-1)\psi_g \rangle =-\eta \langle \psi_g,\Gamma(-1)\psi_g \rangle,
	\end{equation*}
	and $\langle \psi_g,\Gamma(-1)\psi_g \rangle$ converges to 0 for $g$ tending to $\infty$. Furthermore, given any sequence of elements $\{ g_n \}_{n=1}^\infty\subset \RR$ tending to $\infty$ there is a subsequence $\{ g_{n_i}  \}_{i=1}^\infty$ such that
	\begin{equation*}
	\lim_{i\rightarrow \infty}\lvert v(k)\lvert^2 \lVert (F_{-\eta,g_{n_i}}-\cE_{\eta,g_{n_i}}+\omega(k))^{-1}\psi_{g_{n_i}}-\omega(k)^{-1}\psi_{g_{n_i}}\lVert^2=0
	\end{equation*}
	almost everywhere.
\end{lemma}
\begin{proof}
	We have
	\begin{equation*}
	U_gF_{\eta,g}U_g^*+g^2\lVert \omega^{-1/2}v\lVert^2 =d\Gamma(\omega)+\eta W(2g\omega^{-1}v,-1)=\widetilde{F}_{\eta}(2g\omega^{-1}v,\omega).
	\end{equation*}
	Note that 
	\begin{equation*}
	\langle \Omega , \widetilde{F}_{\eta}(2g\omega^{-1}v,\omega)\Omega\rangle=\eta\exp (-2g^2\lVert \omega^{-1}v \lVert^2 )\leq 0
	\end{equation*}
	so $\cE_{\eta,g}+g^2\lVert \omega^{-1/2}v\lVert= \inf(\sigma(\widetilde{F}_{\eta}(2g\omega^{-1}v,\omega)))\leq 0 $. This implies
	\begin{equation*}
	0\leq \langle \widetilde{\psi}_g,d\Gamma(\omega)\widetilde{\psi}_g \rangle  \leq - \eta\langle \widetilde{\psi}_g,W(2g\omega^{-1}v,-1)\widetilde{\psi}_g \rangle=   \lvert \eta\lvert \langle \psi_g,\Gamma(-1)\psi_g \rangle   \leq \lvert \eta \lvert.
	\end{equation*}
	Since $\psi_g\in \cD(N^{1/2})$ by Theorem \ref{BasicpropGS} we find (see \cite{Thomas1})
	\begin{equation*}
	a(k)U_g\psi_g=U_ga(k)\psi_g+gv(k)\omega(k)^{-1}U_g\psi_g,
	\end{equation*}
	and so the pull through formula from equation (\ref{eq:pull1}) gives
	\begin{equation*}
	a(k)\widetilde{\psi}_g=-gv(k)U_g(F_{-\eta ,g}-\cE_{\eta,g}+\omega(k))^{-1}\psi_g+gv(k)\omega(k)^{-1}U_g\psi_g.
	\end{equation*} 
Hence we find
	\begin{align*}
	&\langle \widetilde{\psi}_g,d\Gamma(\omega)\widetilde{\psi}_g \rangle\\&=g^2\int_{\cM}\omega(k)\lvert v(k)\lvert^2 \lVert (F_{-\eta,g}(v,\omega)-\cE_{\eta,g}+\omega(k))^{-1}\psi_{g}-\omega(k)^{-1}\psi_g\lVert^2 dk.
	\end{align*}
	Since this remains bounded by $\lvert \eta \lvert $ as $g$ tends to infinity, we conclude that the integral converges to 0 as $g$ tends to infinity. Thus existence of the desired subsequence follows from standard measure theory. Assume now that the conclusion about convergence of $\langle \psi_g,\Gamma(-1)\psi_g \rangle$ is false. We may then pick $\varepsilon>0$ and sequence $\{ g_n \}_{n=1}^\infty$ such that $- \eta \langle \psi_{g_n},\Gamma(-1)\psi_{g_n} \rangle \geq \varepsilon$ for all $n$ and 
	\begin{equation*}
	\lim\limits_{g\rightarrow \infty}\lvert v(k)\lvert^2 \lVert (F_{-\eta,g_n}-\cE_{\eta,g_n}+\omega(k))^{-1}\psi_{g_n}-\omega(k)^{-1}\psi_{g_n}\lVert^2=0
	\end{equation*}
	for almost every $k\in \RR^\nu$. Let $P_g$ be the spectral measure of $F_{-\eta,g}-\cE_{\eta,g}=F_{-\eta,g}-\cE_{-\eta,g}$ and define the measure $\mu_{g}(A)=\langle \psi_{g} ,P_g(A)\psi_g \rangle$. Since $v\neq 0$ we see
	\begin{align*}
	\lVert (F_{-\eta,g_n}-\cE_{\eta,g_n}&+\omega(k))^{-1}\psi_{g_n}-\omega(k)^{-1}\psi_{g_n}\lVert^2\\&=\int_{[0,\infty)}\left \lvert \frac{1}{\lambda+\omega(k)}-\frac{1}{\omega(k)}\right \lvert^2 d\mu_{g_n}(\lambda)
	\end{align*}
	converges to 0 for some $k\in \RR^\nu$ where $\omega(k)> 0$. Since the integrals above converges to $0$, the numbers $\mu_{g_n}([\varepsilon/2,\infty))$ must converge to 0, as the integrand has a positive lower bound on $[\varepsilon/2,\infty)$. In particular $P_{g_n}([0,\varepsilon/2))\psi_{g_n}-\psi_{g_n}$ will converge to 0. Hence we find for $n$ larger than some $K$ that
	\begin{equation*}
	-\eta\langle P_{g_n}([0,\varepsilon/2))\psi_{g_n},\Gamma(-1)P_{g_n}([0,\varepsilon/2))\psi_{g_n} \rangle\geq \frac{3\varepsilon}{4} \lVert P_{g_n}([0,\varepsilon/2))\psi_{g_n} \lVert^2.
	\end{equation*}
	Let $x_n=P_{g_n}([0,\varepsilon/2))\psi_{g_n}$. By Lemma \ref{Lem:FundamentalIneq} we find $d\Gamma(\omega)+g\varphi(v)+g^2\lVert \omega^{-1/2}v\lVert^2\geq 0$. Using this and  $\cE_{\eta,g}=\cE_{-\eta,g}\leq -g^2\lVert \omega^{-1/2}v\lVert^2$ we may calculate for $n\geq K$
	\begin{align*}
	(\cE_{\eta,g_n}(v,\omega)+\varepsilon/2 )\lVert x_n\lVert^2 &\geq \langle x_n,F_{-\eta,g_n} x_n \rangle\\&=-\eta\langle x_n,\Gamma(-1)x_n \rangle+\cE_{-\eta,g_n}\lVert x_n\lVert^2\\&\,\,\,\,+\langle x_n,(d\Gamma(\omega)+g_n\varphi(v)+g_n^2\lVert \omega^{-1/2}v\lVert^2)x_n\rangle \\&\,\,\,\,-(\cE_{-\eta,g_n}+g_n^2\lVert \omega^{-1/2}v\lVert^2)\lVert x_n\lVert^2\\& \geq -\eta\langle x_n,\Gamma(-1)x_n \rangle+\cE_{-\eta,g_n}\lVert x_n\lVert^2\\& \geq (3\varepsilon/4 +\cE_{\eta,g_n}(v,\omega))\lVert x_n\lVert^2,
	\end{align*}
	which is the desired contradiction.\cqfd
	\end{proof}
\begin{proof}[Proof of Theorem \ref{Mssless-Case}]
	For each $g\geq 0$ we let  $\psi_g$ be a ground state eigenvector for $F_{\eta,g}$. Define $U_g=W(g\omega^{-1}v,1)$ and $\widetilde{\psi}_g=U_g\psi_g$.  We see that
	\begin{align*}
	 \lvert \cE_{\eta,g}+g^2\lVert \omega^{-1/2}v \lVert \lvert &=\lvert \langle \widetilde{\psi}_g,\widetilde{F}_{\eta}(2g\omega^{-1}v,\omega)\widetilde{\psi}_g \rangle \lvert\\&=\lvert  \eta \langle \psi_g,\Gamma(-1)\psi_g \rangle+\langle \widetilde{\psi}_g,d\Gamma(\omega)\widetilde{\psi}_g \rangle  \lvert \\&\leq 2\lvert \eta \lvert \langle \psi_g,\Gamma(-1)  \psi_g \rangle,
	\end{align*}
	which converges to 0 for $g$ tending to $\infty$ by Lemma \ref{CentralResMassless}. It only remains to prove the statement regarding the number operator. Let $\{g_n\}_{n=1}^\infty$ be any sequence converging to $\infty$. Pick a subsequence $\{g_{n_i}\}_{i=1}^\infty$ such that
	\begin{equation*}
	\lim_{i\rightarrow \infty}\lvert v(k)\lvert^2 \lVert (F_{-\eta,g_{n_i}}-\cE_{\eta,g_{n_i}}+\omega(k))^{-1}\psi_{g_{n_i}}-\omega(k)^{-1}\psi_{g_{n_i}}\lVert^2=0
	\end{equation*}
	almost everywhere. Using equation (\ref{eq:pull1}) we see that
	\begin{equation*}
	a(k)\psi_g=-gv(k)(F_{-\eta,g}-\cE_{\eta,g}+\omega(k))^{-1}\psi_g
	\end{equation*} 
	and so
	\begin{align*}
	&\frac{\lvert \langle \psi_{g_{n_i}},N\psi_{g_{n_i}} \rangle-g_{n_i}^2\lVert \omega^{-1}v \lVert^2\lvert }{g^2_{n_i}}\\&\leq  \int_{\cM} \lvert v(k)\lvert^2 \lvert \lVert (F_{-\eta,g_{n_i}}-\cE_{\eta,g_{n_i}} +\omega(k))^{-1}\psi_{g_{n_i}}\lVert^2-\lVert \omega(k)^{-1}\psi_{g_{n_i}}\lVert^2 \lvert dk,
	\end{align*}
	which goes to 0 as $i$ tends to infinity by dominated convergence.\cqfd
\end{proof}

\section{Proof of Theorem \ref{Dim1+2}}
In this section we will assume $\cH=L^2(\RR^\nu,\cB(\RR^\nu),\lambda^{\nu})$ and that $\omega$ is a selfadjoint, non-negative and injective multiplication operator on this space. Then $m_{\textup{ess}}(\omega)=m(\omega):=m$ since $\sigma(\omega)=\sigma_{\textup{ess}}(\omega)$ (See \cite{Thomas1}). Furthermore, we define $P=\lvert \Omega\rangle \langle \Omega \lvert$ and $\overline{P}=1-P$. Then $\overline{ P}$ clearly reduces $d\Gamma(\omega)$ and $\Gamma(-1)$. Let $\overline{d\Gamma}(\omega)$ and $\overline{\Gamma}(-1)$ denote the restrictions to $\overline{\cF_b(\cH)}=\overline{P}\cF_b(\cH)$. For $v\in \cH$ we define $\overline{\varphi}(v)$ as the restriction of $\overline{P}\varphi(v)\overline{P}$ to $\overline{\cF_b(\cH)}$. Note that it is symmetric and infinitesimally $\overline{d\Gamma}(\omega)$ bounded when $v\in \cD(\omega^{-1/2})$. Hence we may define 
\begin{equation*}
\overline{F}_{\eta}(v,\omega)=\overline{d\Gamma}(\omega)+\eta\overline{\Gamma}(-1)+\overline{\varphi}(v),
\end{equation*}
which is selfadjoint on $\cD(\overline{d\Gamma}(\omega))$ and bounded below when $v\in \cD(\omega^{-1/2})$. Note $\inf(\sigma(\overline{F}_{\eta}(v,\omega)))\geq \cE_\eta(v,\omega)$ by the min-max principle. Furthermore, one may repeat the argument for Lemma \ref{POS} to show that for every $\lambda<\cE_{\eta}(v,\omega)$ we have
\begin{equation*}
(\overline{F}_{\eta}(v,\omega)-\lambda)^{-1}\overline{P}\cC_+\subset \overline{P}\cC_+.
\end{equation*}
To summarise
\begin{lemma}\label{RedSpace}
	If $v\in \cD(\omega^{-1/2})$ then $\overline{F}_{\eta}(v,\omega)$ is selfadjoint selfadjoint and bounded below by $\cE_\eta(v,\omega)$. Furthermore $(\overline{F}_{\eta}(v,\omega)-\lambda)^{-1}\overline{P}\cC_+\subset \overline{P}\cC_+$ for every $\lambda<\cE_{\eta}(v,\omega)$.
\end{lemma}
We shall also need the following lemma.
\begin{lemma}\label{Feshbach}
	 For all $\lambda<  \cE_{\eta}(v,\omega)$ we have
	\begin{equation*}
	0<\langle \Omega,(F_{\eta}(v,\omega)-\lambda)^{-1}\Omega \rangle=( \eta-\lambda+\langle v,(\overline{F}_{\eta}(v,\omega)-\lambda)^{-1}v \rangle  )^{-1}.
	\end{equation*}
\end{lemma}
\begin{proof}
	Let $\lambda< \cE_{\eta}(v,\omega)$. One easily checks that $(F_{\eta}(v,\omega)-\lambda,d\Gamma(\omega)+\eta\Gamma(-1)-\lambda)$ is a Feshbach pair for $P$. Write $T=d\Gamma(\omega)+\eta\Gamma(-1)-\lambda$, $H=F_{\eta}(v,\omega)-\lambda$ and $W=H-T=\varphi(v)$. The Feshbach map $F$ is now given by
	\begin{align*}
	F&=PHP-PW\overline{P} (\overline{F}_{\eta}(v,\omega)-\lambda )^{-1}\overline{P}WP\\&=(\eta-\lambda)P+\langle v,(\overline{F}_{\eta}(v,\omega)-\lambda )^{-1}v \rangle P.
	\end{align*}
	This is invertible from $\text{Span}(\Omega)$ to $\text{Span}(\Omega)$ since $H$ is invertible. To calculate the inverse using we use the formula in \cite{Feshbach} and find
	\begin{align*}
	F^{-1}=PH^{-1}P=\langle \Omega,(F_{\eta}(v,\omega)-\lambda)^{-1}\Omega \rangle P.
	\end{align*}
	If one identifies the the linear maps from $\text{Span}(\Omega)$ to $\text{Span}(\Omega)$ with $\CC$ we find the desired equality. Positivity follows since $H^{-1}$ maps $\cC_+$ into $\cC_+$, and we know that the matrix element is not zero since the Feshbach map is invertible. \cqfd
\end{proof}
We may now prove Theorem \ref{Dim1+2}. The basic technique for proving this result comes from the paper \cite{Spohn} where it is used for the translation invariant Nelson model.
\begin{proof}[Proof of Theorem \ref{Dim1+2}]
Let $\eta>0$ and assume the conclusion does not hold. Since $F_{-\eta}(v,\omega)$ has a ground state by Theorem \ref{Thm:Spectral Theory of decomposition} the only option is that $F_{\eta}(v,\omega)$ does not have a ground state. By Theorems \ref{Thm:Spectral Theory of decomposition} and \ref{BasicpropGS} we note that $\cE_{\eta}(v,\omega)=\inf(\sigma_{\textup{ess}}(F_\eta(v,\omega)))= \cE_{-\eta}(v,\omega)+m$ and that $F_{-\eta }(v,\omega)$ has a ground state $\psi$ which has non-zero inner product with $\Omega$. By Lemma \ref{Feshbach} we find 
	\begin{equation*}
	\lambda-\eta> \langle v,(\overline{ F}_{\eta}(v,\omega)-\lambda)^{-1}v \rangle
	\end{equation*}
	for all $\lambda<\cE_{\eta}(v,\omega)=\cE_{-\eta}(v,\omega)+m$, and so $ \langle v,(\overline{F}_{\eta}(v,\omega)-\lambda)^{-1}v \rangle$ is uniformly bounded from above for all $\lambda<\cE_{-\eta}(v,\omega)+m$. We shall now prove that this leads to a contradiction with the assumption in equation (\ref{eq:Assumption}). The following pull through formula, holds for $x\in \cD(d\Gamma(\omega))$ such that $(F_{\eta}(v,\omega)-\lambda  )x\in \cD(N^{1/2})$ (see \cite{Thomas1})
	\begin{align}\label{eq:pullgen}
	a(k)x=&(F_{-\eta}(v,\omega)+\omega(k)-\lambda)^{-1}a(k)(F_{\eta}(v,\omega)-\lambda  )x\\&-v(k)(F_{-\eta}(v,\omega)+\omega(k)-\lambda)^{-1}.\nonumber 
	\end{align}
	We note that
	\begin{equation*}
	(F_{\eta}(v,\omega)-\lambda  )(\overline{F}_{\eta}(v,\omega)-\lambda)^{-1}v=PF_{\eta}(v,\omega)(\overline{F}_{\eta}(v,\omega)-\lambda)^{-1}v+v\in \cD(N^{1/2}).
	\end{equation*}
	Hence we may apply equation $(\ref{eq:pullgen})$ with $x=(\overline{F}_{\eta}(v,\omega)-\lambda)^{-1}v$. Now $a(k)P=0$ so $a(k)(F_{\eta}(v,\omega)-\lambda  )x=v(k)\Omega$. This implies
	\begin{align*}
	\overline{v}(k)a(k)(\overline{F}_{\eta}(v,\omega)-\lambda)^{-1}v&= \lvert v(k)\lvert ^2(F_{-\eta}(v,\omega)+\omega(k)+\lambda )^{-1}\Omega\\&-\lvert v(k)\lvert ^2(F_{-\eta}(v,\omega)+\omega(k)-\lambda)^{-1}(\overline{F}_{\eta}(v,\omega)-\lambda)^{-1}v.
	\end{align*}
	Taking the inner product with $\Omega$, we obtain two terms. Both are non-negative by Lemmas \ref{POS} and \ref{RedSpace} so
	\begin{align*}
	\langle \Omega, v(k)a(k)(\overline{F}_{\eta}(v,\omega)-\lambda)^{-1}v \rangle&\geq \lvert v(k)\lvert ^2\langle \Omega,(F_{-\eta}(v,\omega)+\omega(k)-\lambda)^{-1}\Omega\rangle \\& \geq  \lvert \langle\Omega,\psi\rangle\lvert ^2 \lvert v(k)\lvert^2  (\omega(k)+\cE_{-\eta}(v,\omega)-\lambda).
	\end{align*}
	Hence we find
	\begin{align*}
	\langle v,(\overline{F}_{\eta}(v,\omega)-\lambda)^{-1}v \rangle&=\int_{\cM}  \langle \Omega, v(k)a(k)(\overline{F}_{\eta}(v,\omega)-\lambda)^{-1}v \rangle dk\\&\geq \lvert \langle\Omega,\psi\rangle\lvert ^2\int_{\cM} \lvert v(k)\lvert^2  (\omega(k)+\cE_{-\eta}(v,\omega)-\lambda)dk, 
	\end{align*}
	which goes to infinity for $\lambda$ tending to $\cE_{-\eta}(v,\omega)+m$ by the monotone convergence theorem, equation (\ref{eq:Assumption}) and the fact $ \lvert \langle\Omega,\psi\rangle\lvert ^2\neq 0$. This contradicts the boundedness of $\langle v,(\overline{F}_{\eta}(v,\omega)-\lambda)^{-1}v \rangle$.

	In the special case mentioned, let $\omega(x_0)=m$ be the global minimum of $\omega$. Using Taylor approximations there is $r>0$ such that for $x\in B_r(x_0)$ we have $0\leq \omega(k)-m \leq C\lvert k-x_0\lvert^2$. Switching to polar coordinates yields the result. \cqfd
\end{proof}

\appendix

\section{Various transformation statements.}
In this appendix various useful transformation theorems is stated. Sources are \cite{Thomas1}, \cite{DerezinskiGerard} and \cite{Parthasarathy}. 
\begin{lemma}\label{Lem:SeconduantisedBetweenSPaces}
	Let $U$ be a unitary operator from $\cH$ into some Hilbert space $\cK$. Then there is a unique unitary map $\Gamma(U):\cF_b(\cH)\rightarrow \cF_b(\cK)$ such that $\Gamma(U)\epsilon(g)=\epsilon(Ug)$. If $\omega$ is selfadjoint on $\cH$, $V$ is unitary and $f\in \cH$ then
	\begin{align*}
	\Gamma(U)d\Gamma(\omega)\Gamma(U)^*&=d\Gamma(U\omega U^*).\\
	\Gamma(U)W(f,V)\Gamma(U)^*&=W(Uf,UVU^*).\\
	\Gamma(U)\varphi(f)\Gamma(U)^*&=\varphi(Uf).
	\end{align*}
	Furthermore $\Gamma(U)(f_1\otimes_s\cdots\otimes_s f_n)=Uf_1\otimes_s\cdots\otimes_s Uf_n$ and $U\Omega=\Omega$.
\end{lemma}
One may transform the field operators and second quantised observables by the Weyl transformations. One then obtains the following important statements that we shall need. The proof is an easy calculation using exponential vectors
\begin{lemma}\label{Weyltrans}
	Let $f,h\in \cH$ and $U\in \cU(\cH)$. Then
	\begin{align*}
	W(h,U)\varphi(g)W(h,U)^*&=\varphi(Ug)-2\textup{Re}(\langle Ug,h \rangle)\\
	W(h,U)a(g)W(h,U)^*&=a(Ug)-\langle Ug,h \rangle\\
	W(h,U)a^\dagger(g)W(h,U)^*&=a^\dagger(Ug)-\langle h,Ug \rangle.
	\end{align*}
	Furthermore, if $\omega$ is a selfadjoint, non negative and injective operator on $\cH$ and $h\in \cD( \omega U^* )$ then
	\begin{equation*}
	W(h,U)d\Gamma(\omega)W(h,U)^*=d\Gamma(U\omega U^*)-\varphi(U\omega U^*h)+\langle h,U\omega U^*h\rangle  
	\end{equation*}
	on the domain $\cD(d\Gamma(U\omega U^*))$.
\end{lemma}
In what follows we consider two fixed Hilbert spaces $\cH_1$ and $\cH_2$. We will need the following two lemmas.
\begin{lemma}\label{Iso1}
	There is a unique isomorphism $U:\mathcal{F}_b(\mathcal{H}_1\oplus \mathcal{H}_2)\rightarrow  \mathcal{F}_b(\mathcal{H}_1)\otimes \mathcal{F}_b(\mathcal{H}_2)$ such that $U(\epsilon(f\oplus g))=\epsilon(f)\otimes \epsilon(g)$. The map has the following transformation properties. If $\omega_i$ is selfadjoint on $\mathcal{H}_i$, $V_i$ is unitary on $\mathcal{H}_i$ and $f_i\in \cH_i$ then 
	\begin{align*}
	UW(f_1\oplus f_2,V_1\oplus V_2)U^*&=W(f_1,V_1)\otimes W(f_2,V_2)\\
	Ud\Gamma (\omega_1\oplus \omega_2)U^*&=d\Gamma (\omega_1)\otimes 1 +1\otimes d\Gamma ( \omega_2)\\
	U\varphi(f_1,f_2)U^*&=\varphi(f_1)\otimes 1+1\otimes \varphi(f_2)\\
	Ua(f_1,f_2)U^*&=a(f_1)\otimes 1+1\otimes a(f_2)\\
	Ua^{\dagger}(f_1,f_2)U^*&=a^{\dagger}(f_1)\otimes 1+1\otimes a^{\dagger}(f_2).
	\end{align*}
\end{lemma}

\begin{lemma}\label{Iso2}
	There is a unique isomorphism
	\begin{equation*}
	U:\mathcal{F}(\mathcal{H}_1)\otimes \mathcal{F}(\mathcal{H}_2)\rightarrow \mathcal{F}(\mathcal{H}_1)\oplus \bigoplus_{n=1}^{\infty} \mathcal{F}(\mathcal{H}_1)\otimes S_{n}(\mathcal{H}_2^{\otimes n})
	\end{equation*}
	such that
	\begin{equation*}
	U(w \otimes \{\psi_2^{(n)} \}_{n=0}^\infty)=\psi^{(0)}w\oplus \bigoplus_{n=1}^{\infty} w \otimes \psi_2^{(n)}.
	\end{equation*}
	Let $A$ be a selfadjoint operator on $\mathcal{F}(\mathcal{H}_1)$ and $B$ be selfadjoint on $\mathcal{F}(\mathcal{H}_2)$ such that $B$ is reduced by all of the subspaces $S_{n}(\mathcal{H}_2^{\otimes n})$. Write $B^{(n)}=B\mid_{S_{n}(\mathcal{H}_2^{\otimes n})}$. Then
	\begin{align*}
	U(A\otimes 1+1\otimes B)U^*&=A+B^{(0)}\oplus \bigoplus_{n=1}^\infty( A\otimes 1+1\otimes B^{(n)} )\\ U A\otimes B U^*&=A \otimes B= B^{(0)} A \oplus \bigoplus_{n=1}^\infty A\otimes  B^{(n)}.
	\end{align*}
\end{lemma}
\begin{acknowledgements} Thomas Norman Dam was supported by the Independent Research Fund Denmark
with through the project "Mathematics of Dressed Particles".
\end{acknowledgements}


\begin{thebibliography}{15}
	
\bibitem{Volker}
V. Bach, M. Ballesteros, M. K\"onenberg and L. Menrath. Existence of Ground State Eigenvalues for the Spin-Boson Model with Critical Infrared Divergence and Multiscale Analysis. ArXiv:1605.08348 [math-ph].


\bibitem{Hirokawa1}
Betz V., Hiroshima F., Lorinczi J.: Feynman-Kac-Type Theorems and Gibbs Measures on Path Space, with applications to rigorous Quantum Field Theory. De Gruyter Studies in Mathematics, \textbf{34}. Walter De Gruyter \& CO, Berlin (2011).

\bibitem{Thomas1} Dam, T. N.,M\o ller J. S.: Spin Boson Type Models Analysed Through Symmetries. Arxiv:1803.05812 (Accepted in Kyoto Journal of Mathematics).


\bibitem{DerezinskiGerard}
Derezinski J., G\'erard C.: Asymptotic completeness in quantum field theory. Massive Pauli-Fierz Hamiltonians, Rev. Math. Phys. \textbf{11} 383-450 (1999).

\bibitem{Derezinski}
Derezinski J.: Van Hove Hamiltonians - Exactly Solvable Models of the Infrared and Ultraviolet Problem. Ann. Henri Poncar\'e \textbf{4} 713-738 (2003).

\bibitem{DirkPizzo}
Deckert D. A., Pizzo A.: Ultraviolet Properties of the Spinless, One-Particle Yukawa Model. Commun. Math. Phys. \textbf{327}, 887–920 (2014).

\bibitem{Frlich2}
Fr\"ohlich J.: On the infrared problem in a model of scalar electrons and massless
scalar bosons. Ann. Inst. Henri Poincar´e 19 (1973), 1-103

\bibitem{Gerard}
G\'erard C.: On the existence of ground states for massless Pauli-Fierz Hamiltonians, Ann. Henri Poincar\'e \textbf{1}, 443-459 (2000).

\bibitem{GlimmJaffe}
Glimm J., Jaffe A.: The $\lambda(\phi^4)_2$ quantum field theory without cutoffs: II. The field operators and the approximate vacuum, Ann. Math. 91, 362–401 (1970).

\bibitem{Feshbach}
Griesemer, M., Hasler D.: On the smooth Feshbach-Schur map. Journal of Functional Analysis \textbf{254}, 2329-2335 (2008).

\bibitem{Hasler}
Hasler D., Herbst I.: Ground States in the Spin Boson Model. Annales Henri Poincar\'e Volume \textbf{12}, pp. 621-677 (2011).


\bibitem{Masao}
Hirokawa M.: The Rabi model gives off a flavor of spontaneous SUSY breaking.
 Quantum Studies: Mathematics and Foundations \textbf{2}, 379-388  (2015).


\bibitem{Hirokawa2}
Hirokawa M., Hiroshima F, Lorinczi, J.: Spin-boson model through a Poisson-driven stochastic process. Mathematische Zeitschrift \textbf{277}, 1165–1198 (2014).

\bibitem{MasaoJacob}
Hirokawa M., M\o ller J., Sasaki I.: A Mathematical Analysis of Dressed Photon in Ground State of Generalized Quantum Rabi Model Using Pair Theory. Journal of Physics A: Mathematical and Theoretical \textbf{50}  (2017). 



\bibitem{Merkli}
Merkli M., K\"onenberg M., Song H.: Ergodicity of the spin-boson model for arbitrary coupling strength. Commun. Math. Phys. \textbf{336}, 261-285 (2014).



\bibitem{Farris}
Miyao, T.: Nondegeneracy of ground states in nonrelativistic quantum field theory. J. Operator Theory  \textbf{64}, 207–241  (2010).

\bibitem{MasaoJacob}
Hirokawa M., M\o ller J., Sasaki I.: A Mathematical Analysis of Dressed Photon in Ground State of Generalized Quantum Rabi Model Using Pair Theory. Journal of Physics A: Mathematical and Theoretical \textbf{50}  (2017). 

\bibitem{Moller}
M\o ller J.S.: Fully coupled Pauli-Fierz systems at zero and positive temperature. J. Math. Phys. 55, 075203 (2014).



\bibitem{Parthasarathy} Parthasarathy K.R.: An Introduction to Quantum Stochastic Calculus, Monographs in Mathematics, vol. \textbf{85}, Birkh\"auser, Basel, 1992.

\bibitem{RS1}
Reed M., Simon B.: Methods of Modern Mathematical Physics volume I. Functional Analysis Revised and enlarged edition. Elsevier, Amsterdam (1980).


\bibitem{RS4}
Reed M., Simon B.: Methods of Modern Mathematical Physics IV. Analysis of Operators. Elsevier, Amsterdam (1978).

\bibitem{Marcel}
Roeck W. D.,Griesemer M., Kupiainen A.: Asymptotic completeness for the massless spin-boson model. Advances in Mathematics \textbf{268}, 62-84.


\bibitem{Schmudgen}
Schm\"udgen K.: Unbounded Self-adjoint operators on Hilbert Space. Springer, New York, (2012).

\bibitem{Spohn}
Spohn, H.: The polaron at large total momentum, J. Phys. A 21, 1199–1211 (1988).

\bibitem{Yos}
Yoshihara F., Fuse T., Ashhab S., Kakuyanagi K., Saito S., Semba K.: Superconducting qubit–oscillator circuit beyond the ultrastrong-coupling regime. Nature Physics \textbf{13}, 44–47 (2017).



\bibitem{Weidmann}
Weidmann J.: Linear Operators in Hilbert Spaces. Springer, New York (1980).
\end{thebibliography}
\end{document}